\pdfoutput=1
\documentclass[11pt]{article}

\usepackage[margin=1in]{geometry}
\usepackage{setspace}

\usepackage{amsmath,amssymb,amsthm}
\usepackage{graphicx}
\usepackage{hyperref}
\usepackage[ruled,vlined]{algorithm2e}
\usepackage{refcount}
\usepackage{xcolor}
\usepackage{enumitem}
\usepackage{booktabs}
\usepackage{longtable}
\usepackage{array}
\usepackage{comment}
\usepackage{authblk}
\usepackage{lineno}
\usepackage{fancyhdr}
\usepackage{mathtools}
\usepackage[capitalize,nameinlink,noabbrev]{cleveref}
\crefformat{equation}{#2(#1)#3}
\crefmultiformat{equation}{(#2#1#3)}%
{ and~#2(#1)#3)}{, #2(#1)#3}{ and~#2(#1)#3}
\crefname{assumption}{Assumption}{Assumptions}
\crefname{problem}{Problem}{Problems}
\crefname{claim}{Claim}{Claims}
\crefname{fact}{Fact}{Facts}

\theoremstyle{plain}
\newtheorem{theorem}{Theorem}
\newtheorem{lemma}{Lemma}

\newtheorem{proposition}{Proposition}
\newtheorem{conjecture}{Conjecture}
\newtheorem{claim}{Claim}

\newtheorem{fact}{Fact}
\theoremstyle{definition}
\newtheorem{definition}{Definition}

\theoremstyle{remark}
\newtheorem{remark}{Remark}

\definecolor{darkGreen}{RGB}{0 150 0}

\newcommand{\Helia}[1]{\textcolor{purple}{Helia: #1}}
\newcommand{\Cam}[1]{\textcolor{blue}{Cam: #1}}

\newcommand{\shahin}[1]{\textcolor{darkGreen}{#1}}

\newcommand{\OPT}{\textsc{Opt}}

\newcommand{\diam}{\mathrm{diam}}
\newcommand{\cost}{\text{cost}}
\newcommand{\E}{\mathbb{E}}

\newcommand{\kw}[1]{\textbf{#1}}

  \usepackage{nth}
  \usepackage{intcalc}

  \newcommand{\cAAAI}[1]{AAAI\ Conference\ on\ Artificial (AAAI)}

\title{Fairness in the $k$-Server Problem} 

\author[1]{Mohammadreza Daneshvaramoli}
\author[1]{Helia Karisani}
\author[1]{Mohammad Hajiesmaili}
\author[2]{Shahin Kamali}
\author[1]{Cameron Musco}

\affil[1]{University of Massachusetts Amherst \\
\texttt{mdaneshvaram@umass.edu}, \texttt{hkarisani@umass.edu}, \texttt{hajiesmaili@cs.umass.edu},
\texttt{cmusco@cs.umass.edu}
}

\affil[2]{York University \\
\texttt{kamalis@yorku.ca}
}

\date{}
\begin{document}
\maketitle


\begin{abstract}
We initiate a formal study of fairness for the $k$-server problem, where the objective is not only to minimize the total movement cost, but also to distribute the cost equitably among servers. We first define a general notion of $(\alpha,\beta)$-fairness, where, for parameters $\alpha \ge 1$ and $\beta \ge 0$, no server incurs more than an $\alpha/k$-fraction of the total cost plus an additive term $\beta$. We then show that fairness can be achieved without a loss in competitiveness in both the offline and online settings. In the offline setting, we give a deterministic algorithm that, for any $\varepsilon > 0$, transforms any optimal solution into an $(\alpha,\beta)$-fair solution for $\alpha = 1 + \varepsilon$ and $\beta = O(\mathrm{diam} \cdot \log k / \varepsilon)$, while increasing the cost of the solution by just an additive $O(\mathrm{diam} \cdot k \log k / \varepsilon)$ term. Here $\mathrm{diam}$ is the diameter of the underlying metric space. We give a similar result in the online setting, showing that any competitive algorithm can be transformed into a randomized online algorithm that is fair with high probability against an oblivious adversary and still competitive up to a small loss.

The above results leave open a significant question: can fairness be achieved in the online setting, either with a deterministic algorithm or a randomized algorithm, against a fully adaptive adversary? We make progress towards answering this question, showing that the classic deterministic Double Coverage Algorithm (DCA) is fair on line metrics and on tree metrics when $k = 2$. However, we also show a negative result: DCA fails to be fair for any non-vacuous parameters on general tree metrics. We further show that on uniform metrics (i.e., the paging problem), the deterministic First-In First-Out (FIFO) algorithm is fair. We show that any ``marking algorithm'', including the Least Recently Used (LRU) algorithm, also satisfies a weaker, but still meaningful notion of fairness.
\end{abstract}




\section{Introduction}
\label{sec:intro}



The $k$-server problem, introduced by Manasse, McGeoch, and Sleator~\cite{MMS88}, is central to the theory of online algorithms. In this problem, $k$ servers move within a metric space to serve a sequence of requests, and the goal is to minimize the total movement cost. In particular, the algorithm is presented with an input sequence $\sigma = (\sigma_1,\ldots,\sigma_T)$ consisting of $T$ requests, each corresponding to a point in the metric space. At time $t \in \{1,\ldots,T\}$, it must move some server to request $\sigma_t$, incurring a cost equal to the distance moved by that server. The challenge is to decide which server to move at each step in order to minimize the total movement cost. In the online setting, such a decision is made without prior knowledge of forthcoming requests. 

Over the past three decades, the $k$-server problem has served as a fundamental benchmark for studying online decision-making and competitive analysis, with deep connections to paging, metrical task systems, and beyond~\cite{fiat1991competitive,BLS92,BorodinElYanivBook, Kou09, Fiat1994, predrag2020}. Many variants on the problem, including those where server costs are weighted or determined by different metric spaces \cite{BansalEKN23,AyyadevaraC21,AyyadevaraCS24,KoutsoupiasT04,ChiplunkarV20,BansalEKN23}, or where requests have both pick-up and drop-off locations \cite{Coester2019, Buchbinder2020,SODA26Taxi} have been studied.
The problem has also been investigated under advice complexity, where the algorithm receives some information about the input in the form of some error-free bits of advice~\cite{EmekFKR11,RenaultR15,BockenhauerKKK17,GuptaKL16}, and more recently under learning-augmented frameworks with predictions, where additional information is provided, but may be erroneous~\cite{LindermayrMS22,AntoniadisCEPS23,AntoniadisICML23,ChristiansonSW23}.


\smallskip

\noindent\textbf{Fairness for $k$-server.} The $k$-server problem can be motivated by scenarios where a service provider 
dispatches mobile agents to client locations, aiming to minimize the total travel distance. 
In such settings, it is also important to consider a notion of fairness: ideally, the workload should be balanced to avoid overburdening any single server.
Similarly, in distributed systems and cloud platforms, $k$ virtual machines (servers) offering identical services may need to migrate between nodes to handle incoming requests~\cite{CastenowFKMH22, ananthanarayanan2013greening, delimitrou2014quasar}. In such domains, in addition to minimizing the total migration cost, it is important to ensure that this cost is distributed fairly among servers. 

Fairness has been studied in a wide range of offline and online decision-making problems. Examples include various forms of clustering \cite{AhmadianE0M19,MakarychevV21,KnittelSDH23,Dickerson24}, matching \cite{hajiaghayi2024fairness,hosseini2023class,Aigner-HorevS22}, scheduling \cite{Ajtai1998,NiuTV23,SchedulingSODA25, bonifaci2006}, selection problems like knapsack \cite{Patel0L21,LechowiczS0KH24,FluschnikSTW19}, and the secretary problem \cite{KoppenVT12,BalkanskiMM24}. A common interpretation of fairness involves ensuring proportional balance across different groups or agents. For instance, fair clustering methods may consider colored points and seek to minimize the standard clustering cost while ensuring that all colors are approximately equally represented in each cluster~\cite{AhmadianE0M19,KnittelSDH23}. Similarly, fair class matching considers settings where vertices are associated with agents, aiming to maximize the number of matched vertices while ensuring that each agent is proportionally represented in the matching~\cite{hosseini2023class,hajiaghayi2024fairness,AISTAT25Fairness}.

More closely related to $k$-server, Singh et al.~\cite{Singh2023} study the \emph{$k$-FOOD problem}, where requests are defined by source-destination pairs and associated pickup time windows, effectively enforcing deadlines for serving a request. In the standard version of the problem, the objective is to minimize the total distance traveled by all servers, while in the \emph{fair} variant, the goal is to minimize the maximum distance traveled by any server.  \cite{Singh2023}  establishes NP-hardness of the offline version of both the standard and fair versions -- hardness stems primarily from additional constraints such as time windows and deadlines. Note that these constraints are not present in the classical $k$-server formulation, and the offline version of the problem is solvable in polynomial time~\cite{BorodinElYanivBook}.

Martinez-Sykora et al.~\cite{MartinezSykoraMCF24} study a food delivery problem with fairness constraints that aim to minimize the range of waiting times across servers -- they formulate this problem as an ILP and demonstrate that it can be solved on small input instances, but do not show that it is tractable in general. Ong et al.~\cite{OngPYGS24} study algorithms for ride scheduling that aim to ensure fairness to customers (i.e., requests rather than servers) -- they focus primarily on experimental evaluation rather than theoretical guarantees. Chiplunkar et al.~\cite{Chiplunkar2023} study online competitive algorithms for a min-max fairness objective for $k$-paging, which is a special case of $k$-server, again focusing on fairness across requests rather than servers.


Despite the vast literature on fairness in algorithm design, including on problems related to $k$-server, perhaps surprisingly, little work has directly addressed fairness in the $k$-server problem itself. The original work by Manasse, McGeoch, and Sleator~\cite{MMS88} proposed the \emph{Balance Algorithm (BAL)}, which selects a server to move at each step by greedily minimizing the maximum cost incurred across all servers after that step. This algorithm is explicitly motivated by fairness and was shown to be $k$-competitive on metric spaces with $k+1$ points~\cite{MMS88}. Unfortunately, however, 
BAL fails to achieve a bounded competitive ratio on general metrics. 
Thus, to the best of our knowledge, the question of designing algorithms for the $k$-server problem that are both fair and competitive has remained entirely open.

The well-studied \emph{paging problem} is a classical special case of the $k$-server problem in which the underlying metric space is uniform—i.e., all pairwise distances are equal~\cite{SleatorT85,karlin1988competitive}. Paging models the behavior of caches in operating systems, memory hierarchies, and web caching. A \emph{fault} occurs when a request arrives for a page (vertex) that is not currently stored in one of the $k$ cache slots (servers). In the $k$-server interpretation, each cache slot corresponds to a server, and each fault corresponds to moving a server to the requested vertex at unit cost. Classical deterministic paging algorithms include \textsc{First-In-First-Out} (FIFO) and \textsc{Least-Recently-Used} (LRU). In this interpretation, FIFO moves the server that has not moved for the longest time, while LRU moves the server whose most recent movement occurred farthest in the past. Both FIFO and LRU achieve the optimal deterministic competitive ratio of $k$ for paging~\cite{SleatorT85}. This close connection between paging and the $k$-server problem motivates our later analysis of the fairness properties of FIFO, LRU, and marking algorithms on uniform metrics.

\subsection{Our Contributions}

In this work, we initiate a formal study of fairness as an objective for the $k$-server problem. Our contributions are summarized below.

\subsubsection{\texorpdfstring{$(\alpha,\beta)$}{(alpha,beta)}-Fairness Definition}
\label{sect:fairnessDef}

We start by formalizing a natural notion of fairness for the $k$-server problem: we say that a {deterministic} algorithm $A$ is $(\alpha, \beta)$-fair on request sequence $\sigma$ if no server pays cost greater than $\alpha \cdot \frac{\text{cost}(A,\sigma)}{k} + \beta$, where $\text{cost}(A,\sigma)$ is the total cost paid by all servers to serve $\sigma$. A perfectly fair algorithm, in which all servers pay identical costs, would have $\alpha = 1$ and $\beta = 0$.

For randomized algorithms, we distinguish between \emph{ex-ante} and \emph{ex-post}  fairness~\cite{AzizFSV24}. Ex-ante fairness means \emph{fairness in expectation} -- i.e., 
a {randomized} algorithm $A$ is $(\alpha, \beta)$-ex-ante fair on request sequence $\sigma$ if no server has an \emph{expected cost} greater than $\alpha \cdot \frac{\E[\text{cost}(A,\sigma)]}{k} + \beta$. 
Via a simple random swapping strategy, it is easy to achieve $(1, \diam)$-ex-ante fairness while maintaining competitiveness. Here $\diam$ denotes the diameter of the underlying metric space.
Suppose we start with an algorithm over the given metric space, with expected cost $\E[\cost(A,\sigma)]$ on input $\sigma$. At the beginning of the algorithm, we randomly permute the server identities, swap the server positions according to the permutation, and reassign the workloads of the servers going forward according to the permutation as well. Then, after the additive $\le \diam$ cost paid per server for the initial swap, by symmetry, all servers have exactly the same \emph{expected cost}, equal to $\frac{\E[\text{cost}(A,\sigma)]}{k}$. 

The above approach, however, is not very compelling. While the servers have the same {expected cost} up to a small additive factor after the initial random swap, the cost distribution of the servers may still be very unfair for any given choice of swap. 
Thus, throughout the paper, we adopt an \emph{ex-post} notion of fairness for randomized algorithms. Specifically, we say that a randomized algorithm $A$ is $(\alpha,\beta)$-fair on a request sequence $\sigma$ if, with high probability over its random choices, every server incurs cost no larger than $\alpha \cdot \frac{\text{cost}(A,\sigma)}{k} + \beta$. The exact probability for which the bounds holds  will naturally impact the achievable $\alpha$ and $\beta$.



While we focus on $(\alpha,\beta)$-fairness in this work, we show that it either encompasses or can be easily reduced to other natural notions of fairness, e.g., multiplicative fairness in which all servers pay cost within a fixed multiplicative factor of each other, or additive fairness, in which all servers pay cost within a fixed additive factor of each other. We discuss such alternative definitions along with reductions and connections between them in the full version of the paper.

\paragraph{Egalitarian cost model.}
While our main results focus on the standard utilitarian objective (total movement cost),
our fairness guarantees also have implications for the \emph{egalitarian} cost model,
where the cost of an algorithm is defined as the maximum distance traveled by any server.
We defer the formal definitions, bounds, and proofs for this model to
Appendix~\ref{app:egalitarian}.

\subsubsection{Offline Fairness}

Our first contribution is to show that, in the offline setting, $(\alpha,\beta)$-fairness is achievable without significant loss in competitiveness. In particular, in Section \ref{sec:offline} we prove:
\begin{theorem}[Theorem \ref{thm:fairness} Restated]\label{thm:intro1}
For any $k$-server input sequence $\sigma$ over a metric space with diameter $\diam$ and any $\varepsilon > 0$, there is an offline solution for $\sigma$ that is $(\alpha, \beta)$-fair for $\alpha = 1 + \varepsilon$ and
$
\beta = O\left( \frac{\log k \cdot \diam}{\varepsilon} \right),
$ and has total cost  bounded by $\OPT(\sigma) +  O \left (\frac{k \log k \cdot \diam}{\varepsilon}\right ),$ where $\OPT(\sigma)$ is the offline optimal  cost. Further, this fair solution can be computed in polynomial time.
\end{theorem}
That is, we can achieve fairness up to an $\alpha = 1+\varepsilon$ factor while paying just a small additive factor $\beta$ and a small increase in the optimal offline cost, which is a function of $k$, $\diam$ and $1/\varepsilon$, but importantly, independent of the input length and the offline optimal cost. 

\smallskip

\noindent\textbf{Proof overview.}
To prove Theorem \ref{thm:intro1}, we introduce an algorithm (Algorithm \ref{alg:fair_cost}) that takes any optimal solution for an input sequence $\sigma$ and transforms it into a fair solution through a sequence of pairwise server \emph{swaps}. A pairwise swap is a modification to a solution sequence in which, at a given time $z$, two servers exchange positions in the metric space, and at all times after $z$, they serve the requests previously served by the other server. 

Our algorithm proceeds in rounds, iteratively improving the fairness of the initial solution. At each round, it identifies the server $H$ that currently pays the highest cost and the server $L$ that currently pays the lowest cost. It further identifies a time step such that, if the servers are {swapped} at this point, their loads will be nearly balanced. We prove that, as long as the solution is not yet $(\alpha,\beta)$-fair, after any server participates in at least two swaps, the maximum server load is decreased by a multiplicative factor $ \approx 1+\varepsilon$. Thus, after $O(\log_{1+\varepsilon} k) = O(\frac{\log k}{\varepsilon})$ swaps per server, the maximum load is decreased from potentially as high as $k$ times the average load to at most $(1+\varepsilon)$ times the average load. Since each swap itself incurs additive cost $\le 2 \cdot \diam$, this leads to our bound of $\alpha = 1+\varepsilon$, $\beta = O\left (\frac{\log k \cdot \diam}{\varepsilon}\right)$ and our increase in cost over $\OPT(\sigma)$ of $O\left (\frac{k\log k \cdot \diam}{\varepsilon}\right)$.

It is not hard to see that our swapping algorithm runs in polynomial time, and thus, since an offline optimal solution for $k$-server can be computed in polynomial time~\cite{BorodinElYanivBook}, $(\alpha,\beta)$-fairness is achievable in polynomial time. This contrasts with recent work on related fair scheduling problems, which have proven to be computationally hard~\cite{Singh2023,MartinezSykoraMCF24}. 

\paragraph{On the dependence on diam.}
Some of our fairness guarantees include an additive term instantiated using a parameter
$\diam$, which denotes the maximum separation ever created between any two servers
during the algorithm’s execution (and not the diameter of the underlying metric space).
Under the standard assumption that all servers start at the same location, $\diam$ is
always bounded by the total movement cost incurred by the algorithm.

In particular, for executions in which $\diam \le \gamma \cdot \mathrm{cost}(A,\sigma)$
for some $\gamma \in (0,1)$, additive $\diam$-dependent terms in our bounds can be
re-expressed as multiplicative slack.
For example, a guarantee of the form
\[
c_i(A,\sigma) \le \frac{\alpha}{k}\,\mathrm{cost}(A,\sigma) + \beta
\]
will become
\[
c_i(A,\sigma) \le
\left(\frac{\alpha}{k} + \gamma \cdot \frac{\beta}{\diam}\right)\mathrm{cost}(A,\sigma).
\]
Interpreted this way, our results remain meaningful even on executions where the metric
diameter is large but the total movement cost dominates.

\subsubsection{Online Fairness using Randomization}

Our second contribution shows that $(\alpha,\beta)$-fairness can also be achieved for online $k$-server  using a randomized algorithm against an oblivious adversary --i.e., where the request sequence $\sigma$ is chosen ahead of time and may not depend on the decision of the algorithm (this is the standard setting that randomized  algorithms for online $k$-server are studied). In particular, in Section \ref{sec:online} we prove:

\begin{theorem}[Theorem \ref{thm:randomized} Restated]\label{thm:intro2}
For any fixed $\gamma > 0$ and $\varepsilon > 0$, there is a randomized online $k$-server algorithm that, on any fixed  input sequence $\sigma$ over a metric space with diameter $\diam$  achieves  $(\alpha,\beta)$-fairness for $\alpha = 1+\varepsilon$ and $\beta = O(c(\sigma)^{1/(1+\gamma)} \cdot \diam)$ with probability $ 1 - k\cdot \exp\left( -  \Omega \left (\frac{\varepsilon^2}{\gamma\cdot k}\cdot c(\sigma)^{1/(\gamma+1)} \right ) \right) $. Further, the total cost incurred by the algorithm is at most $c(\sigma) + O(k\cdot c(\sigma)^{1/(1+\gamma)} \cdot \diam)$. Here, $c(\sigma)$ is the best-known cost achievable by a randomized algorithm over the given metric space, e.g., $c(\sigma) = \mathrm{polylog}(k,n) \cdot \OPT(\sigma)$ for $n$-point metrics~\cite{bansal2015polylogarithmic,bubeck2018k}.
\end{theorem}
To interpret Theorem \ref{thm:intro2}, consider the simplified setting with $\gamma = 1$ and $\varepsilon = 1$. The theorem shows that $(\alpha,\beta)$-fairness is achievable with $\alpha = 2$, $\beta = O(c(\sigma)^{1/2} \cdot \diam)$, and cost increase $O(k \cdot c(\sigma)^{1/2} \cdot \diam)$ with probability $1-k \cdot \exp(-\Omega(c(\sigma)^{1/2})$. For sequences where $c(\sigma)$ is large, this probability is very close to $1$, and both $\beta$ and the cost increase are  $o(c(\sigma))$, and thus negligible compared to the total cost. Thus,
Theorem \ref{thm:intro2} is analogous to Theorem \ref{thm:intro1}, establishing that in the online setting, strong fairness is achievable without significantly impacting competitiveness. The parameter $\gamma$, which does not appear in the offline result, balances a tradeoff between $\beta$ and the cost increase against the success probability.


\smallskip

\noindent\textbf{Proof overview.}
To prove Theorem~\ref{thm:intro2}, we use the randomized server-swapping technique, in a similar way to the approach described in Section~\ref{sect:fairnessDef} for achieving ex-ante (i.e., expected) fairness. 
%
However, to give an algorithm that is fair with high probability on any input sequence,
we make a series of swaps instead of one, randomly permuting the servers at each swap. We argue via concentration inequalities that the servers have similar total loads with high probability. 

 The difficulty is determining how often to swap -- too many swaps increase the cost of the original algorithm too much, but too few swaps mean that fairness is not guaranteed with high probability. To balance this trade-off, we introduce a phase-based algorithm (Algorithm \ref{alg:randomized_fairness}) that swaps anytime the total cost incurred increases by a large enough factor. As an example, in the case when $\gamma =1$, the algorithm swaps after incurring a total cost of $1$, then after incurring an additional total cost of $2$, then after incurring an additional total cost of $3$, and so on. It makes $s$ swaps after incurring total cost roughly $\sum_{i=1}^s s = \Omega(s^2)$. Thus, after incurring cost $c(\sigma)$, the algorithm has made $O(c(\sigma)^{1/2})$ swaps, and incurs total cost at most $O(c(\sigma)^{1/2})$ between each swap. This enables us to argue that both the additional additive cost associated with the swaps and the imbalance in server loads are bounded by $O(c(\sigma)^{1/2})$ with high probability and thus negligible in comparison to the total cost $c(\sigma)$. We trade off the success probability and the cost increase through the parameter $\gamma$ in  Theorem~\ref{thm:intro2} -- roughly choosing phase $i$ with cost increase $i^\gamma$ so that the total cost incurred after $s$ swaps is $\Omega(s^{1+\gamma})$. Increasing $\gamma$ leads to fewer swaps, but a lower probably of successfully satisfying the fairness guarantee.
     
  \subsubsection{Towards Deterministic Online Fairness}
  
  Theorems \ref{thm:intro1} and \ref{thm:intro2} establish that fairness is achievable for the $k$-server problem in both the offline setting and the online setting using randomization against an oblivious adversary. It is straightforward to verify that Theorem \ref{thm:intro2} also holds for an adaptive online adversary that can change its input based on the server locations at each round but is unable to distinguish servers based on their IDs (and thus, cannot tell if a random swap of the servers has occurred). However, it remains open if fairness can be achieved against stronger online adversaries or, potentially, deterministically. For example, the following is open:

\begin{center}
\textit{ Is there a deterministic online $k$-server algorithm that is $O(k)$-competitive on general metrics while also being $(\alpha,\beta)$-fair for $\alpha = O(1)$ and $\beta$ depending only on $k$ and $\diam$?}
\end{center}

Towards progress on the above question, in Section  \ref{sec:classic}, we investigate fairness properties of several widely studied deterministic online $k$-server algorithms that are competitive on certain classes of metrics. In particular, we consider fairness of the Double Coverage Algorithm (DCA), which is known to be $k$-competitive on tree metrics. We also study the fairness of several $k$-competitive deterministic algorithms that apply to uniform metrics -- i.e., to the classic paging problem. These include the First-In First-Out (FIFO) algorithm, and any \emph{marking algorithm}, encompassing e.g., the Least Recently Used (LRU) algorithm.

\begin{table}[]
    \centering
\scalebox{.9}{
    \begin{tabular}{|c|c|c|c|c|}
    \hline
      Metric   & Algorithm & Competitive Ratio & Fairness & Reference  \\
      \hline
       Line  & DCA & $k$~\cite{ChrobakL91} & $(1,O(k\cdot \diam))$ & Theorem~\ref{thm:intro3} \\
        Tree with $k=2$ & DCA & $k(=2)$~\cite{ChrobakL91} & $(1.5,1.75\cdot \diam)$ & Theorem~\ref{thm:intro4} \\ 
        Uniform Metrics (paging) & FIFO & $k$~\cite{SleatorT85} & (1,1) & Theorem~\ref{th:FIFOOrg} \\
        Metrics with $k+1$ points & \textsc{Balance} & $k$~\cite{MMS88}& $(1,\diam)$ & Proposition~\ref{prop:balance} \\
        \hline
    \end{tabular}
    }
    \caption{Summary of deterministic fair online algorithms for different metrics}
    \label{tab:placeholder}
\end{table}

\paragraph{Tree Metrics.}
We first show two positive results for DCA. We prove that the algorithm is strongly fair for line metrics (i.e., tree metrics whose underlying graph is a line). In particular:
\begin{theorem}[Theorem \ref{thm:additive-fairness} Restated]\label{thm:intro3} The DCA algorithm run on any input sequence $\sigma$ over a line metric with diameter $\diam$ gives a solution in which the absolute difference between the cost paid by any two servers is $O(k \cdot \diam)$.
\end{theorem}


Theorem \ref{thm:intro3} implies that DCA is $(1,O(k \cdot \diam))$-fair on line metrics. The proof is fairly straightforward: on a line, servers can be ordered from left to right. We can readily establish that no algorithm aiming to minimize total cost (including DCA) will ever change their relative ordering. Thus, each server moves left and right, always serving requests that land between itself and its adjacent left and right neighbors (or one of the endpoints of the line). DCA couples the movements of adjacent servers: whenever a request lands between two servers, both move towards it. Due to this coupling, we can argue that the cost difference between neighboring servers depends only on their net movement, which is bounded by $\diam$. Telescoping across all $k$-servers yields a maximum overall gap in costs of $O(k \cdot \diam)$.

We also prove that DCA is $(\alpha,\beta)$-fair for the $2$-server problem on general tree metrics, for $\alpha = O(1)$ and $\beta = O(\diam)$. In particular:
\begin{theorem}[Theorem \ref{thm:multiplicatively_fair} Restated]\label{thm:intro4} The DCA algorithm run on any input sequence $\sigma$ over a tree metric with diameter $\diam$ for the $2$-server problem is $(1.5,1.75\cdot \diam)$-fair. 
\end{theorem}

We prove Theorem \ref{thm:intro4} by splitting the total cost paid by each server into the \emph{diverging cost}, when the server is moving away from the other server, and the \emph{converging cost}, when the server is moving towards the other. We show in Lemma \ref{lem:DCAdivConv} that these costs both comprise a constant fraction of the total cost paid by the server. Since DCA couples converging costs (when one server moves towards the other, the other server moves as well), this allows us to argue that the servers pay the same costs up to roughly a constant multiplicative factor.

Despite the above positive results, we show that, surprisingly, DCA can exhibit substantial unfairness on general tree metrics,  forcing one server to do a majority of the total work.
\begin{theorem}[Theorem \ref{thm:DCA_additive_fair} Restated]\label{thm:intro5}
There is a tree metric and an input sequence $\sigma$ over this metric on which DCA outputs a solution in which a single server pays cost $\Omega(k \cdot \OPT(\sigma))$, where $\OPT(\sigma)$ is the offline optimal cost for request sequence $\sigma$.
\end{theorem}
Note that DCA is known to be $k$-competitive on general tree metrics and thus it incurs total cost $\le k \cdot \OPT(\sigma)$. Theorem \ref{thm:intro5} shows that a constant fraction of this cost can be incurred by just a single server, and thus DCA is not $(\alpha,\beta)$-fair for any $\alpha = o(k)$ and fixed $\beta$. 

The hard case is a tree metric that is similar to a line: the tree consists of a line of \emph{major} nodes, each of which connects to a very close by \emph{minor} leaf node. Requests alternate between midpoints on the line and nearby minor nodes, repeatedly forcing the leftmost server to traverse the entire structure, past other servers which have been moved to minor nodes so that they can be `passed'. This causes the leftmost server to incur cost $\Omega(k \cdot \OPT)$, despite DCA being $k$-competitive overall.

\paragraph{Uniform Metrics.}
We next turn our attention to uniform metrics -- where all points have the same distance from each other. This variant of $k$-server is also know as the \emph{paging problem}. We prove several positive results for existing algorithms. For simplicity, assume that the metric is scaled so all points are at distance $1$ from each other. First, we prove that the classic First-In First-Out (FIFO) algorithm, which is known to be $k$-competitive on uniform metrics~\cite{SleatorT85}, achieves essentially the strongest possible notion of fairness: since the algorithm simply `cycles through' servers (i.e., cache positions) to serve any incoming request that is not already covered (i.e., to cover any cache miss), all servers pay the same cost, up to an additive factor of one. Formally,
\begin{theorem}[Theorem \ref{thm:fifo} Restated]
    FIFO is $(1,1)$-fair. \label{th:FIFOOrg}
\end{theorem}

We next observe (Theorem \ref{thm:lru}) that the popular Least Recently Used (LRU) algorithm does not enjoy such a strong fairness guarantee: there are input instances where one server pays essentially all of the cost of the algorithm (i.e., one cache slot is reloaded for essentially all cache misses seen). However, LRU, and in fact any of the broader class of \emph{marking algorithms}, still achieves an interesting notion of fairness: no server pays cost greater than $\OPT(\sigma)$ on request sequence $\sigma$. I.e., no server pays more than what the average cost per server might be on a worst-case input sequence with competitive ratio $k$. Formally:
\begin{theorem}[Theorem \ref{thm:marking} Restated]\label{thm:intro_marking}
For any marking algorithm for the $k$-server problem on uniform metrics, the cost incurred by any given server on input $\sigma$ is at most $\OPT(\sigma)$. 
\end{theorem}
This result contrasts e.g., with Theorem \ref{thm:intro5}, which shows that for DCA on general tree metrics, one server may pay up to $k$ times the optimal offline cost.
The proof follows a similar approach to how marking algorithms are proven to be $k$-competitive: we argue that between any two cache misses that cause the same cache position to be reloaded, at least $k$ unique files must be requested. Thus, the sequence containing these two misses must cause at least one miss even for the optimal offline solution. Iterating this argument, no cache position (i.e., server) has more faults than $\OPT(\sigma)$ over the full input sequence. 




\paragraph{Metrics with $k+1$ Points.} Finally, we consider metrics with $k+1$ points, which have been widely studied in the literature on the $k$-server problem. It is known that no deterministic algorithm can achieve a competitive ratio better than $k$ on any metric with $k+1$ points~\cite{MMS88}. It is also well known that the \textsc{Balance} algorithm matches this optimal competitive ratio~\cite{MMS88}. Recall that this algorithm is inherently fair:  it always moves the server whose total cost, after potentially serving the current request, is minimum among all servers. We can therefore conclude the following:

\begin{proposition}
\textsc{Balance} is a $(1,\diam)$-fair algorithm for metrics with $k+1$ points. \label{prop:balance}
\end{proposition}

     \subsection{Roadmap}

The remainder of the paper is organized as follows. Section~\ref{sec:setup} introduces the $k$-server problem and formalizes our $(\alpha, \beta)$-fairness notion. We explore alternative notions of fairness in Appendix~\ref{sec:beyond}. 
Section~\ref{sec:offline} presents our offline $(\alpha,\beta$)-fair algorithm and Section~\ref{sec:online} presents our randomized online algorithm. Section~\ref{sec:classic} evaluates the fairness of deterministic algorithms for particular metric spaces, such as DCA, FIFO, and  LRU. We conclude the paper with open questions in Section~\ref{sec:conclusion}. 

\section{Problem Setup and \texorpdfstring{$(\alpha,\beta)$}{(alpha,beta)}-Fairness}
\label{sec:setup}

In this section, we formally define the $k$-server problem and introduce notation used throughout the paper. We then define the central fairness notion that we consider.


\smallskip

\noindent\textbf{The $k$-server problem:} The $k$-server problem is defined over a metric space $(M, d)$ with $k$ servers, labeled $1,\ldots,k$, initially located at positions $s_1(0), \dots, s_k(0) \in M$. At each time step $t = 1, 2, \dots, T$, a request $\sigma_t \in M$ arrives. The algorithm must choose one server to move to the requested location, incurring a cost equal to the distance moved -- i.e., $d(s_i(t-1),\sigma_t)$ if server $i$ is moved and was previously at position $s_i(t-1)$. The total cost is the sum of the costs over the full sequence of requests $\sigma = (\sigma_1, \dots, \sigma_T)$.

\smallskip

\noindent \textbf{Cost notation:} Let $A$ be a $k$-server algorithm. For a fixed request sequence $\sigma$, let $c_i(A,\sigma,t)$ denote the distance that server $i$ moves at time $t$ under algorithm $A$, and define
$
c_i(A,\sigma) = \sum_{t=1}^T c_i(A,\sigma,t)$ 
as the total cost incurred by server $i$. The total cost of the algorithm is $
\text{cost}(A, \sigma) = \sum_{i=1}^k c_i(A,\sigma).$

We denote by $\OPT(\sigma)$ the minimum total cost achievable by any offline algorithm on input $\sigma$.
When $A$ and $\sigma$ are clear from context, we will sometimes drop them as arguments, using $c_i(t)$, $c_i$, $\text{cost}$, and $OPT$ to refer to the above quantities. 

\smallskip

\noindent \textbf{Competitive ratio:} An algorithm $A$ is said to be \emph{$c$-competitive} if there exists a constant $C \ge 0$ (possibly depending on $k$ and on the underlying metric space) such that for all input sequences $\sigma$,
$
\text{cost}(A, \sigma) \leq c \cdot \OPT(\sigma) + C.$ 
For randomized algorithms, we consider the \emph{expected} cost and say that $A$ is \emph{$c$-competitive in expectation} if
$
\mathbb{E}[\text{cost}(A, \sigma)] \leq c \cdot \OPT(\sigma) + C.
$

The $k$-server problem is typically studied in the \emph{online setting}, where an algorithm must make its decision at time $t$ after seeing request $\sigma_t$, but before seeing any later requests. It is well known that no deterministic online algorithm can achieve competitive ratio $< k$ for general metrics~\cite{KP95}. This lower bound is matched on tree metrics by the well-known Double Coverage Algorithm (DCA)~\cite{CL91b}, and up to a constant factor by the Work Function Algorithm (WFA), which is $2k-1$ competitive on general metric spaces~\cite{KP95, koutsoupias1994phd}, and $k$-competitive on line metrics, star metrics, and metric spaces with $k+2$ points~\cite{Bartal2004}. For randomized algorithms against an oblivious adversary, the best known upper bounds are polylogarithmic in $k$ and $n$~\cite{bansal2015polylogarithmic,bubeck2018k}. Recently, the long-standing conjecture that $O(\log k)$ is achievable was disproved; the randomized competitive ratio is now known to be $\Omega(\log^2 k)$ in some metrics~\cite{Bubeck2023}. 






\subsection{\texorpdfstring{$(\alpha,\beta)$}{(alpha,beta)}-Fairness}

We now define our central fairness notion, which requires that no individual server incurs significantly more than its proportional share of the total cost.


\begin{definition}[$(\alpha, \beta)$-Fairness]
\label{def:fairness}
Let $A$ be a $k$-server algorithm, and let $w(\sigma)$ be a reference cost baseline associated with request sequence $\sigma$. We say that $A$ is \emph{$(\alpha, \beta)$-fair with respect to $w$} if for all input sequences $\sigma$ and all servers $i \in [k]$,
\begin{equation}
c_i(A, \sigma) \leq \frac{\alpha \cdot w(\sigma)}{k} + \beta,
\label{eq:fairness_general}
\end{equation}
where $c_i(A, \sigma)$ denotes the total cost incurred by server $i$ under algorithm $A$ on input $\sigma$.

\medskip
\textbf{Randomized case.}
A randomized algorithm $A$ is \emph{$(\alpha, \beta)$-fair with probability $1-\delta$ with respect to $w$} if for every input sequence $\sigma$,
\begin{equation}
\Pr\left[\, \forall i \in [k],\;
c_i(A, \sigma) \le \frac{\alpha \cdot w(\sigma)}{k} + \beta \,\right]
\ge 1-\delta.
\end{equation}
When $\delta \le 1/\mathrm{poly}(k,|\sigma|)$, we simply say that $A$ is $(\alpha,\beta)$-fair \emph{with high probability}.
\end{definition}

Note that our general definition allows fairness to be expressed relative to a chosen baseline $w$, such as the optimal offline cost $\OPT(\sigma)$ or the algorithm's own total cost, $\text{cost}(A, \sigma)$. 

\begin{remark}[Fairness v.s. Competitiveness]
We remark that when $w(\sigma) = \text{cost}(A, \sigma)$, $(\alpha,\beta)$-fairness is not an interesting notion unless the algorithm $A$ is also required to have a bounded competitive ratio. Otherwise, we could have all servers move in synchrony and achieve perfect fairness, but at the cost of competitiveness. Roughly speaking, in the online setting, given an algorithm $A$ that is $\alpha$ competitive, and assuming that each time a server moves to serve request $\sigma_t$ we move all other servers an equal amount (but do not change their positions), then we have a perfectly fair algorithm that is $\alpha \cdot k$-competitive. Thus, we will always look for fair algorithms that beat this trivial baseline.
\end{remark}


\section{A Fair Offline Algorithm}
\label{sec:offline}

We first show that in the offline setting, we can obtain a solution that is both near-optimal in cost and provably fair across servers. In particular, we design a transformation (Algorithm~\ref{alg:fair_cost}) that takes any (possibly unfair) optimal offline solution and converts it into a solution that is $( 1 + \varepsilon, \beta)$-fair for $\beta = O(\diam \cdot \log k / \varepsilon)$. The conversion increases the total cost by, at most an additive constant.
The key idea is to redistribute cost among servers using a sequence of pairwise swaps that gradually balance the load. The transformation is simple, deterministic, and black-box—it does not require any specific structure of the underlying optimal solution. 

\begin{theorem}[$( 1 + \varepsilon,\beta)$-Fair Transformation for Offline $k$-Server]
\label{thm:fairness}
Consider any $k$-server input sequence 
$\sigma$ of length $T$ over a metric space of diameter $\diam$. Let $\{{s}_i(t)\}_{i \in [k], t \in [T]}$ be a sequence of server positions that achieves optimal cost on this input sequence. 

For any $\varepsilon > 0$, Algorithm~\ref{alg:fair_cost} takes $\{{s}_i(t)\}_{i \in [k], t \in [T]}$ as input and returns a modified sequence of positions $\{\hat {s}_i(t)\}_{i \in [k], t \in [T]}$ that is $( 1 + \varepsilon, \beta)$-fair with respect to $\OPT(\sigma)$, for 
$
\beta =  O\left( \frac{\log k \cdot \diam}{\varepsilon} \right)$. 
Moreover, the total cost of the output solution is $\OPT(\sigma) +  O \left (\frac{k \log k \cdot \diam}{\varepsilon}\right ).$
\end{theorem}

\begin{proof}[Proof Sketch]
Algorithm \ref{alg:fair_cost} iteratively reduces the maximum server cost in the input solution by performing a series of swaps between the maximum and minimum cost servers. Each swap exchanges the servers' remaining work from a chosen time $z$ onward. Using a balancing lemma (Lemma~\ref{lem:l1}), we show that there is always a choice of a swap that reduces the cost difference between the swapped servers to at most an additive $\diam$ factor. These swaps each add at most $2\cdot \diam$ to the total cost of the new algorithm. Further, note that if the solution is not yet fair, the maximum server load is at least $\approx 1+\varepsilon$ times as large as the minimum. This allows us to argue that after a swap, both servers' costs are less than the previous maximum cost by roughly a $1/(1+O(\varepsilon))$ multiplicative factor (Claim~\ref{claim1}). We also observe that the maximum server load always decreases with a swap and that non-swapped servers retain their cost (Claim~\ref{claim2}).
We proceed by considering any window of $k+1$ swaps. By the pigeonhole principle, at least one server must appear twice as the maximum-cost server in a swap. This server’s cost must have decreased by a $\approx 1/(1 + O(\varepsilon))$ factor between these two swaps. Thus, we can argue that the maximum cost drops by a factor of at least $(1 + \varepsilon)$ every $k+1$ rounds (Claim~\ref{claim3}). Repeating this process for $O \left (k \cdot \log_{1+\varepsilon} k\right ) = O\left( \frac{ k \log k}{\varepsilon} \right)$ swaps in order to reduce a potential initial multiplicative imbalance for $k$ to $(1+\varepsilon)$ yields a solution satisfying $(1 + \varepsilon, \beta)$-fairness. See Appendix~\ref{apx:fairness} for full details.
\end{proof}

\begin{algorithm}[t]\LinesNumbered
\caption{Offline Fair $k$-server via Pairwise Swapping}
\label{alg:fair_cost}

\KwIn{Optimal server positions $\{s_i(t)\}_{i \in [k], t \in [T]}$ with corresponding cost allocation $\{{c}_i(t)\}_{i \in [k], t \in [T]}$ and parameter $\varepsilon$}
\KwOut{Modified sequence of positions $\{\hat s_i(t)\}_{i \in [k], t \in [T]}$ satisfying $(1+\varepsilon, \beta)$-fairness.}

Set round counter $r \gets 0$ \label{line:init-round} \\
Initialize $\hat s_i^{(0)}(t) \gets {s}_i(t)$ and $\hat c_i^{(0)}(t) \gets {c}_i(t)$ for all $i \in [k], t \in [T]$ \\
$\beta \gets2 \cdot (1+\varepsilon) \cdot diam \cdot (\frac{3}{2}+\log_{\frac{2+2\varepsilon}{2+\varepsilon}} k)$\\

\While{$\max_{i \in [k]} \sum_{t=1}^T \hat c_i^{(r)}(t) > \frac{ (1 + \varepsilon) \cdot \OPT(\sigma)}{k} + \beta$}{
  Let $H_r \gets \arg\max_{i \in [k]} \sum_{t=1}^T \hat c_i^{(r)}(t)$ \label{line:mnmx} \tcp*{Heaviest-loaded server}
  Let $L_r \gets \arg\min_{i \in [k]} \sum_{t=1}^T \hat c_i^{(r)}(t)$ \tcp*{Lightest-loaded server}

  Let $z \in [T]$ be the any index such that the difference in the total costs of $H_r$ and $L_r$ after a swap at time $z$ become less than diam  \tcp*{Such a $z$ always exists (Lemma~\ref{lem:l1})}

  Initiate current cost/locations for all $i, t$: $\hat c_i^{(r+1)}(t) \gets \hat c_i^{(r)}(t), \hat s_{i}^{(r+1)}(t) \gets \hat s_{i}^{(r)}(t)$\label{line:copy}

  \For{$t > z$}{ 
    $\hat s_{H_r}^{(r+1)}(t) \gets \hat s_{L_r}^{(r)}(t) , $ $\hat s_{L_r}^{(r+1)}(t) \gets \hat s_{H_r}^{(r)}(t)$  \tcp*{swap positions after $z$}
    $\hat c_{H_r}^{(r+1)}(t) \gets \hat c_{L_r}^{(r)}(t) , $ $\hat c_{L_r}^{(r+1)}(t) \gets \hat c_{H_r}^{(r)}(t)$ \label{line:swapc}\tcp*{swap costs after $z$}  
  }
  $\hat c_{H_r}^{(r+1)}(z+1) \gets  d(\hat s_{H_{r}}^{(r+1)}(z+1), \hat s_{H_{r}}^{(r+1)}(z))$ \tcp*{swap penalty}
  $\hat c_{L_r}^{(r+1)}(z+1) \gets   d(\hat s_{L_{r}}^{(r+1)}(z+1), \hat s_{L_{r}}^{(r+1)}(z))$ \tcp*{swap penalty}

  Update round: $r \gets r + 1$
}
\Return{$\{\hat c_i^{(r)}(t)\}$ and $\{\hat s_i^{(r)}(t)\}$}

\end{algorithm}

\section{A Fair Randomized Online Algorithm} 
\label{sec:online}

We next show how to achieve fairness with minimal loss in competitiveness in the online setting with a randomized algorithm. 
Similar to Algorithm~\ref{alg:fair_cost}, we assume access to some existing online $k$-server algorithm $A$ with bounded competitive ratio, but potentially with no fairness guarantees. We use a series of random server swaps to balance the servers' loads and make the algorithm fair. 

At each swap point, we randomly permute all server identifiers, swap the server positions according to the permutation, and reassign the workloads of the servers going forward according to the permutation as well. Thus, between swaps, by symmetry, the expected cost incurred by all servers is the same. Via concentration inequalities, as long as we swap sufficiently often, we can prove that all servers pay similar costs with high probability, establishing $(\alpha,\beta)$-fairness.
The key challenge is determining how often to swap -- too many swaps will increase the cost of the original algorithm too much, but too few swaps will mean that fairness is not guaranteed with high probability. 
To balance this trade-off, we introduce a phase-based algorithm (Algorithm \ref{alg:randomized_fairness}) that swaps anytime the total cost incurred has increased by a large enough factor. We prove the following:

\begin{algorithm}[t]\LinesNumbered
\caption{Online Fair $k$-Server via Phased Random Swapping}
\label{alg:randomized_fairness}
\KwIn{Online request sequence $\sigma = (\sigma_1, \sigma_2, \dots)$, $k$-server algorithm $A$, initial server configuration $(s_1(0),\ldots,s_k(0))$, phase exponent $\gamma > 0$}
Initialize:  phase number $\ell \gets 1$, phase budget $\phi_\ell \gets \ell^\gamma$, and cumulative cost $U_\ell \gets 0$. \\
Reposition servers via a random permutation $\pi_\ell$: for all $i \in [k]$, set $s_i(0) \gets s_{\pi_\ell(i)}(0)$ where $\pi_\ell$ is chosen uniformly at random.

\ForEach{request $\sigma_t$}{
    Serve $\sigma_t$ using algorithm $A$ on servers labeled before repoistion.$\{s_{\pi^{-1}_\ell(i)}(t)\}.$
    \\
    Update $U_\ell \gets U_\ell +cost(A,t)$.\label{line:update_u}

    \If{$U_\ell \geq \phi_\ell$}{
        $\ell \gets \ell + 1$ \label{line:phase_update}, $\phi_\ell \gets \ell^\gamma$, and $U_\ell \gets 0$. \\
        Reposition servers via a random permutation $\pi_\ell$: for all $i \in [k]$, set $s_i(t) \gets s_{\pi_\ell(i)}(t)$ where $\pi_\ell$ is chosen uniformly at random .\label{line:permute}
    }
}
\end{algorithm}

\begin{theorem}[Randomized Online Fairness Guarantee]
\label{thm:randomized}
Let $A$ be any online $k$-server algorithm and let $\sigma$ be any request sequence. Then, Algorithm~\ref{alg:randomized_fairness} with input $A$ and  parameter $\gamma > 0$ satisfies: for any $\varepsilon > 0$, with probability at least $1-  k \cdot \exp\left( - \Omega\left(\frac{\varepsilon^2}{\gamma \cdot k} \cdot  \text{cost}(A, \sigma)^{1/(\gamma+1)} \right) \right),$ each server incurs cost at most
    \[
    (1 + \varepsilon) \cdot \frac{\text{cost}(A, \sigma)}{k} + O\left( \text{cost}(A, \sigma)^{1/(\gamma+1)} \cdot \diam \right).
    \]
    Further, the total cost across all servers is bounded by $\text{cost}(A, \sigma) + O\left( \text{cost}(A, \sigma)^{1/(\gamma+1)} \cdot k \cdot \diam \right)$.

%
\end{theorem}
In terms of $(\alpha,\beta)$-fairness, Theorem \ref{thm:randomized} implies that Algorithm \ref{alg:randomized_fairness} is $(\alpha,\beta)$-fair  with respect to $\text{cost}(A,\sigma)$ for $\alpha = 1+\varepsilon$, $\beta = O\left( \text{cost}(A, \sigma)^{1/(\gamma+1)} \cdot \diam \right)$ with high probability. It pays total cost penalty $O\left( \text{cost}(A, \sigma)^{1/(\gamma+1)} \cdot k \cdot \diam \right)$. Both $\beta$ and this total cost increase grow sublinearly in $\text{cost}(A, \sigma)$, and can be regarded as lower order. 
 The success probability depends on $k$, $\gamma$, and $\varepsilon$, but tends to one as $\text{cost}(A,\sigma)$ tends to infinity.


\begin{proof}[Proof Sketch]
Algorithm~\ref{alg:randomized_fairness} divides the input sequence into phases, with each phase ending once the base algorithm $A$ accumulates cost at least $\phi_\ell = \ell^\gamma$ (line~\ref{line:phase_update}). At the start of each phase (line~\ref{line:permute}), the algorithm applies a random permutation to the server identities and exchanges their positions. This guarantees that the expected share of the phase cost is uniformly distributed across servers, regardless of their prior history.

Let $m$ denote the number of phases. Since each phase contributes $\phi_\ell = \ell^\gamma$ cost and $\sum_{\ell=1}^m \ell^\gamma = \Theta(m^{\gamma+1})$, we have: $m = \Theta\left( \text{cost}(A)^{1/(\gamma + 1)} \right).$
Thus, each server's expected cost is:
\[
\E[c_i] = \frac{\text{cost}(A)}{k} + O\left( m \cdot \diam \right) = \frac{\text{cost}(A)}{k} + O\left( \text{cost}(A)^{1/(\gamma+1)} \cdot \diam \right).
\]

To show high-probability fairness, we apply a concentration inequality to the random variables equal to each server’s cost per phase. Each of these variables is upper-bounded by $\phi_\ell$, with maximum phase cost $\phi_m = m^\gamma = \Theta\left(\text{cost}(A)^{\gamma/(\gamma+1)}\right)$. Thus, we can bound the total variance of these random variables by  $\sigma^2 = O\left( \sum \phi_\ell^2 / k \right )$. 

 Applying Bernstein’s inequality with these upper bounds gives that each server’s total cost remains $(1+\varepsilon)$ close to its expectation with probability at least $1-  \exp\left( -\Omega \left ( \frac{\varepsilon^2}{\gamma\cdot k} \cdot \text{cost}(A, \sigma)^{1/(\gamma+1)} \right) \right) $. A union bound over all $k$ servers gives the final high-probability guarantee: with probability at least $1-  k \cdot \exp\left( - \frac{\varepsilon^2}{\gamma \cdot k} \cdot \Omega\left( \text{cost}(A, \sigma)^{1/(\gamma+1)} \right) \right)$, for all $i \in [k]$, 
\[
c_i \leq (1 + \varepsilon) \cdot \frac{\text{cost}(A)}{k} + O\left( \text{cost}(A)^{1/(\gamma+1)} \cdot \diam \right).
\]
See Appendix~\ref{apx:randomized} for full details of the proof.
\end{proof}


\begin{remark}[Extension to ID-Oblivious Adversary]
We remark that Theorem \ref{thm:randomized} also holds when $\sigma$ is generated by an \emph{ID-oblivious} adaptive adversary -- i.e., an adversary that can pick $\sigma_t$ based on the positions of the servers at rounds $t' < t$, but not on their identities. At the start of each phase, Algorithm \ref{alg:randomized_fairness} applies a random permutation to the server IDs without changing their overall set of positions. As a result, an ID-oblivious adversary observes the same sequence of positions under Algorithm \ref{alg:randomized_fairness}  as under the base algorithm $A$, and thus, the analysis for the oblivious adversary case carries through unchanged.
\end{remark}

\section{Towards Deterministic Online Fairness}
\label{sec:classic}


A central open question left by our randomized fairness result is whether fairness can
be achieved by \emph{deterministic} online algorithms, or against fully adaptive
adversaries that observe server identities. In this section we study the fairness
properties of classical deterministic $k$-server algorithms on specific metric classes.
Our focus is on the Double Coverage Algorithm (DCA) on tree metrics and on
deterministic paging algorithms (FIFO, LRU, and marking algorithms) on uniform
metrics. These algorithms are well-studied from a competitive-analysis perspective,
and analyzing their fairness properties gives insight into what kinds of fairness may be
achievable deterministically.

\subsection{Fairness of the DCA on Special Cases}

The Double Coverage Algorithm (DCA) works on a tree metric and moves all `unblocked' servers that do not have another server on their path to the request point, each at a uniform speed towards the request until the request is served.

We first show that, owing to the inherent fairness of moving all unblocked servers at the same time, the algorithm is fair on line metrics for any $k$ and general tree metrics for $k = 2$. However, we also show a strong negative result: DCA can be highly unfair on general tree metrics for general $k$. 

We begin with our positive result showing that DCA is $(1, O(k \cdot \diam))$-fair with respect to $\OPT$ on line metrics for any value of $k$.

\begin{theorem}
\label{thm:additive-fairness}
Let $\sigma$ be any request sequence over a line metric with diameter $\diam$, and let $c_i$ denote the total cost incurred by server $i$ when serving $\sigma$ using DCA. We have $\max_{i,j \in [k]} |c_i-c_j| =  O(k \cdot \diam)$. Thus, DCA is $(1, O(k \cdot \diam))$-fair on line metrics. 
\end{theorem}

\begin{proof}[Proof Sketch]
Since DCA never moves servers past each other on the line, we can fix a unique labeling of the servers as $1, \ldots, k$ based on their initial positions from left to right. The double coverage approach ensures that, for each request, the two closest servers move toward the request at equal speed until one of them reaches it. Therefore, every time server $i$ moves to the right, it is matched by an equal leftward movement by server $i+1$, and vice versa. This gives the recurrence:
\[
R_i = L_{i+1}, \quad \text{where } R_i \text{ is server $i$'s rightward movement and } L_i \text{ is its leftward movement}.
\]
The total cost for server $i$ is $c_i = R_i + L_i$. Letting $d_i = R_i - L_i$ be the net movement of server $i$,
\[
c_{i+1} - c_i = L_{i+1} + R_{i+1} - (L_i + R_{i}) = R_{i+1} - L_i = d_{i+1} + d_i.
\]
Since all movement is along a line of diameter $\diam$, we have $|d_i| \leq \diam$ for all $i$. Therefore, the cost difference between adjacent servers is bounded by:
\[
|c_{i+1} - c_i| \leq |d_i| + |d_{i+1}| \leq 2 \cdot \diam.
\]
By summing the pairwise differences across all $k-1$ server pairs, we conclude that:
\[
\text{for any $i, j \in [k]$ } |c_i - c_j| \leq 2(k - 1) \cdot \diam,
\]
as claimed. See Appendix~\ref{apx:dca_on_line} for a full proof.
\end{proof}

We next show that DCA is $(\alpha,\beta)$-fair for $\alpha = O(1)$ and $\beta = O(\diam)$ for any tree metric in the special case of $k=2$.



\begin{theorem}
\label{thm:multiplicatively_fair}
    DCA is $(1.5,1.75\cdot\diam)$-fair for any tree metric when $k=2$.
\end{theorem}

\begin{proof}

Fix an input $\sigma$. For each server, let $con_i(\sigma)$ denote the total distance moved by $i$ while ``converging" towards the other server (moving $i$ closer to the other server). Similarly, let $div_i(\sigma)$ denote the total distance moved by $i$ while ``diverging" away from the other server. For example, when a request is made on the shortest path between $i$ and $i'$, the distance moved by any of the two servers contributes to $con_i(\sigma)$. When $i$ moves away from the other server (and the other server does not move), the distance moved by $i$ contributes to $div_i(\sigma)$. 

Intuitively, for two servers the divergence and convergence terms directly
reflect how their mutual distance changes over time: when they diverge the
servers move farther apart, and when they converge they move closer together.
Since their separation can never exceed the diameter of the metric (nor become
negative), the total amount by which divergence can exceed convergence, or vice
versa, is necessarily bounded by $\diam$.

Formally, for any input sequence $\sigma$, with $k = 2$ and any server $i$, the following inequality holds (see Lemma~\ref{lemma:DCATwo} in Appendix~\ref{apx:dcadivconv} for the proof). 
\begin{equation}
-\diam \leq div_i(\sigma) + div_{i'}(\sigma) - 2con_i(\sigma) \leq \diam.   \label{eq:diam}
\end{equation}

Since each unit move of a server $i$ either increases or decreases its distance to the other server, we can write:
$c_i = con_i(\sigma) + div_i(\sigma)$.
%
Furthermore, by the definition of DCA and the nature of converging moves, we have $con_i(\sigma) = con_{i'}(\sigma)$ since both servers move towards each other at the same pace during converging moves (this argument fails for $k > 2$, when one server can be blocked from moving by some third server).
Therefore, using the inequalities in (\ref{eq:diam}) 
we can conclude that:
\begin{align*}
    c_i = con_i(\sigma) + div_i(\sigma) &\leq 3\cdot con_i(\sigma) - div_{i'}(\sigma) + \diam\\ 
    &\leq 3\cdot con_{i}(\sigma) + \diam \\
    &=  1.5\cdot con_{i}(\sigma) + 1.5\cdot con_{i}(\sigma) + \diam\\
    &\leq 1.5\cdot con_{i}(\sigma) + 0.75\cdot(div_i(\sigma) + div_{i'}(\sigma)) + 1.75\cdot\diam.
\end{align*}
Substituting the definitions of \(c_i\) and \(c_{i'}\) and the fact that $con_i(\sigma) = con_{i'}(\sigma)$ into the previous bound yields:
\begin{align*}
c_i\leq 0.75 \cdot c_i + 0.75\cdot c_{i'} + 1.75\cdot\diam  \leq 1.5\cdot\frac{cost(DCA,\sigma)}{2}  + 1.75\cdot\diam.
\end{align*}
\end{proof}


\subsection{Unfairness of DCA on General Tree Metrics}


We next show that, unfortunately, the positive results of Theorems \ref{thm:additive-fairness} and \ref{thm:multiplicatively_fair} cannot be extended to the full range of settings where DCA is $k$-competitive. In particular, DCA can fail to satisfying any reasonable notion of fairness on general tree metrics for general $k$, assigning the majoring of work -- in fact, $\Omega(k \cdot \OPT(\sigma))$ work -- to just a single server. 
%

\begin{theorem}
\label{thm:DCA_additive_fair}
There exists a tree metric and a request sequence $\sigma$ on that metric such that, when the Double Coverage Algorithm (DCA) is run on $\sigma$, some server incurs cost $\Omega(k \cdot \OPT(\sigma))$, where $\OPT(\sigma)$ denotes the cost of the offline optimal solution.
\end{theorem}
\begin{proof}[Proof Sketch]
We construct a tree where $k+1$ major nodes form a path (the “spine”), and each major node has an attached minor leaf at distance $\varepsilon$, for arbitrarily small $\varepsilon$. Label the spine nodes $\{1,2,\ldots, k+1\}$ and place servers, labeled $s_1,s_2,\ldots,s_k$, initially at positions $1, 3, 4, \dots, k+1$ respectively, skipping  $2$ to create a gap. The goal is to  force $s_1$ to move from the start to the end of the spine, incurring cost $\Omega(k)$, while all other servers (and the offline optimal solution) incur cost just $O(1)$. 
The hard request sequence proceeds in repeated steps: 
\begin{itemize}
    \item First, issue a request at position $2$, causing both servers $s_1$ and $s_2$ at positions $1$ and $3$ to move toward it, perturbing the distances by an arbitrarily small amount so that $s_2$ arrives first and actually serves the request.
    \item Next, request the minor leaf connected to $2$. This forces $s_2$ to move off the spine, freeing up position $2$. $s_1$ is `blocked' by $s_2$ and does not move.
    \item Now, issue a request at position $3$; $s_1$ at position $2$ and $s_3$ at position $4$ move toward it, and again the other server (now $s_3$) serves the request and is sent off the spine.
\end{itemize}

This pattern repeats: server $s_1$ slowly moves right, passing each server one by one while the others are diverted off the main path. After $k$ steps, $s_1$ incurs cost $\Omega(k)$ while all other servers incur cost just $O(1+\varepsilon)$. See Figure \ref{fig:dcaLowerBound} for an illustration.

Further, the offline optimum can serve all spine requests using just one server move of length $1$ ($s_1$ moves to position $2$ in the first step), and handle minor leaf requests with $\varepsilon$-moves. Thus, $\text{OPT}(\sigma) = O(1+ \varepsilon k)$, and so, setting $\varepsilon$ arbitrarily small, $s_1$ pays costs $\Omega(k \cdot \text{OPT}(\sigma))$, as desired. See Appendix~\ref{apx:DCA_not_fair} for full details, and an argument that this hard sequence can be repeated to apply to settings where $\OPT(\sigma)$ is arbitrarily large in comparison to $k$ and the metric space diameter.
\end{proof}
\begin{figure}[h]
    \centering
    \includegraphics[width=0.5\linewidth]{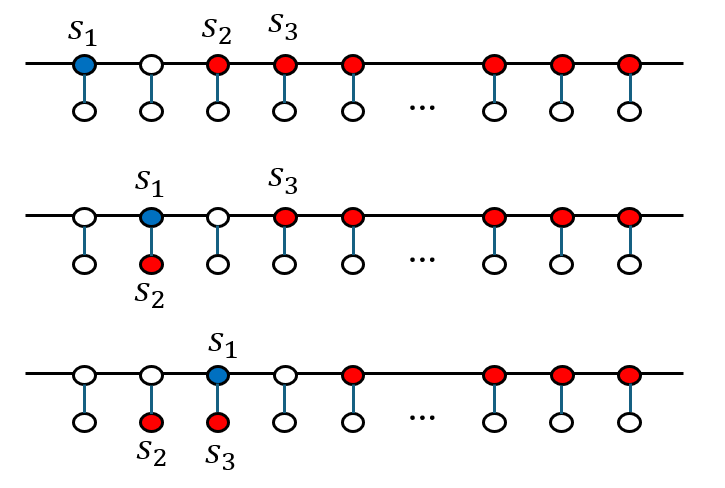}
    \caption{
Hard instance for Theorem~12.
\emph{Top:} Initialization of server positions on the spine.
\emph{Middle:} Requests~1--2 move $s_1$ one step to the right and push $s_2$ off the spine via a
minor leaf.
\emph{Bottom:} Requests~3--4 repeat the process with $s_3$, and so on, forcing $s_1$
to traverse the entire spine while all other servers incur $O(1)$ additional cost.
}
    \label{fig:dcaLowerBound}
\end{figure}

\subsection{Uniform Metrics and Paging}
\label{sec:metrics}


Finally, we consider  the fairness of deterministic $k$-server algorithms on uniform metrics, where we scale all distances to be $1$ for simplicity. This setting is equivalent to the classical paging problem, where the $k$ servers correspond to the $k$ memory cells of a cache. A \emph{fault} occurs when a requested page is not currently in the cache, equivalently, when a request is made to a node (page) without a server located on it. The objective of minimizing total server movements thus becomes equivalent to minimizing the number of faults.

Given this equivalence, we adopt the terminology of paging. To simplify the discussion, we restrict our attention to \emph{lazy paging algorithms}, which evict a page only when a fault occurs and the cache is full. Non-lazy algorithms, such as Flush-When-Full, are also relevant in the $k$-server context,

We start with a strong positive result concerning the First-In-First-Out (FIFO) algorithm, which evicts the page that has been in the cache the longest. It is well known that FIFO is $k$-competitive~\cite{SleatorT85}. We now show that FIFO also guarantees the strongest notion of fairness.

\begin{theorem}\label{thm:fifo}
    FIFO is $(1,1)$-fair on uniform metrics. 
\end{theorem}
\begin{proof}
    Let $A$ and $B$ be two servers (memory cells), and let $a_i, b_i$ denote the indices of the $i$th faults in the input at $A$ and $B$, respectively. Suppose the first fault at $A$ occurs before the first fault at $B$. By the FIFO rule, for any $i$, we have
    \[
        a_i < b_i < a_{i+1} < b_{i+1}.
    \]
    This implies that the number of faults at $A$ exceeds that of $B$ by at most one. Applying this logic to all pairs of servers yields the $(1,1)$-fairness guarantee.
\end{proof}
    
Thus, FIFO is not only optimally competitive (among deterministic algorithms) but also optimally fair. One may ask whether other paging algorithms, such as Least-Recently-Used (LRU), which is also $k$-competitive~\cite{SleatorT85}, share this fairness property. The following theorem shows that they do not. Recall that LRU evicts the page in the cache that has been least recently used.

\begin{theorem}\label{thm:lru}
There exists a request sequence $\sigma$ on a uniform metric such that under LRU,
one server incurs cost $\Omega(\mathrm{cost}(\mathrm{LRU},\sigma))$ while all other servers incur
only $O(1)$. Consequently, LRU does not satisfy $(\alpha,\beta)$-fairness with respect
to its own total cost for any $\alpha = o(k)$ and any $\beta = o(\mathrm{cost}(\mathrm{LRU},\sigma))$.
\end{theorem}

\begin{proof}
We construct an input where a single server incurs (almost) all the cost incurred by LRU.  
    Consider the input sequence
    \[
        (\alpha, 1,2,\ldots,k-1, \beta, 1,2,\ldots,k-1)^m .
    \]
    By the LRU rule, pages $1,2,\ldots,k-1$ remain in the cache throughout at positions $2,\ldots,k$ (i.e., covered by servers $s_2,\ldots,s_k$). The first cache location (server $s_1$) experiences faults at requests to $\beta$ and after repeating the sequence at requests to $\alpha$. 
    Thus, all of the algorithm's cost beyond the initial $k-1$ ``cold misses" (the faults before the cache becomes full) is borne by this single server.
\end{proof}

Despite the above negative result, one can show that LRU, and more generally all \emph{marking algorithms}, guarantee a weaker but still meaningful fairness property: each server incurs a cost no more than the optimal offline solution. That is, they are $(1,0)$-fair with respect to $k \cdot \OPT(\sigma)$, the potential highest cost paid by the algorithm on a worst-case input where it is $k$-competitive. We explore this notion further in Appendix \ref{sec:beyond}, where we call it \emph{acceptable fairness} (Definition \ref{def:acceptable-fairness}).

A marking algorithm works as follows
On a hit, the requested page is simply marked. On a fault, if all pages are already
marked, all marks are cleared; then an unmarked page is evicted and the requested
page is inserted and marked. LRU fits this framework because its notion of
``least-recently-used'' exactly corresponds to the unmarked pages: a page loses its
mark precisely when it has not been requested among the most recent $k$ distinct
requests, matching the classical marking definition.

\begin{theorem}\label{thm:marking}
For any marking algorithm on a uniform metric, the cost incurred by each server on input sequence $\sigma$ is at most $\mathrm{OPT}(\sigma)$. 
\end{theorem}

\begin{proof}
    We use the standard phase-partitioning technique. 
    Fixing the input sequence $\sigma$, define the first phase as the maximal prefix of $\sigma$ containing $k$ distinct pages. Define subsequent phases recursively on the remainder of the sequence. Let $m_\sigma$ denote the number of phases. Then
    $\mathrm{OPT}(\sigma) \geq m_\sigma$,
    since each phase introduces $k+1$ distinct pages and thus forces at least one fault for $\mathrm{OPT}$. On the other hand, in each phase, any server of a marking algorithm incurs at most one fault: after its first fault in a phase, its page is marked and remains protected for the rest of the phase. Thus, each server pays cost at most $\OPT(\sigma)$.
\end{proof}

\section{Conclusion}
\label{sec:conclusion}
In this work, we initiated a formal study of fairness in the $k$-server problem, introducing
$(\alpha,\beta)$-fairness as a unifying framework for capturing equitable cost distribution
among servers.
Within this framework, we showed that strong fairness guarantees can be achieved without
significantly compromising competitiveness.
In particular, we presented an offline transformation that achieves near-perfect fairness
while incurring only a controlled increase in total cost, and a randomized online algorithm
that achieves fairness with high probability against an oblivious adversary.

Beyond algorithmic design, we analyzed the fairness behavior of several classical online
algorithms.
Our results demonstrate that, despite their optimal or near-optimal competitive ratios,
algorithms such as  DCA can exhibit substantial unfairness on general metric spaces,
highlighting an inherent tension between competitiveness and fairness in the online setting.

Several open questions remain.
Can fairness guarantees comparable to our randomized results be achieved by deterministic
online algorithms on general metrics, or against fully adaptive adversaries?
Can the dependence of fairness bounds on structural parameters be further improved?
More broadly, it would also be interesting to consider $(\alpha,\beta)$-fairness for other important generalizations, e.g., metrical task systems~\cite{BLS92}.  
We hope that this work provides a foundation for further investigation of fairness as a
first-class objective in online algorithm design.

\section*{Acknowledgements}

This research was supported by the National Science Foundation under grants
2045641, 2325956, 2512128, and 2533814, and by the Natural Sciences and Engineering Research
Council of Canada (NSERC) under funding reference number RGPIN-2025-07295.





\bibliographystyle{alpha}
\bibliography{references}

\appendix
\label{appendix}
\newpage
\section{Beyond \texorpdfstring{$(\alpha,\beta)$}{(alpha,beta)}-Fairness: Alternative Fairness Notions}
\label{sec:beyond}


Our $(\alpha,\beta)$-fairness definition captures a natural proportionality condition: no server pays significantly more than its fair share of the total cost, up to a small slack. However, there are many other reasonable ways to formalize fairness.
In this section, we introduce and compare three alternative fairness notions—additive, multiplicative, and acceptable fairness -- some of which have appeared implicitly or in domain-specific contexts~\cite{MMS88}. 

We study how these notions relate to each other, and whether algorithms satisfying one notion can be transformed to satisfy another. 
We start with the basic definitions. For simplicity, throughout this section, we consider deterministic algorithms. Although all definitions can be applied to applied to randomized algorithms, analogously to our $(\alpha,\beta)$-fairness notion.

\begin{definition}[Additive Fairness]
\label{def:add}
An algorithm \( A \) is said to be \emph{$\beta$-additively fair} for $\beta \ge 0$ if, for any input sequence \( \sigma \) we have:
\[
\max_{i,j \in [k]} |c_i(A,\sigma) - c_j(A,\sigma)| \le \beta.
\]
\end{definition}

\begin{definition}[Multiplicative Fairness]
An algorithm \( A \) is said to be \emph{$\alpha$-multiplicatively fair} for $\alpha \ge 1$ if, for any input sequence \( \sigma \) we have:
\[
 \max_{i,j \in [k]} \frac{c_i(A,\sigma)}{c_j(A,\sigma)} \leq \alpha.
\]
\end{definition}

\begin{definition}[Acceptable Fairness]\label{def:acceptable-fairness}
An algorithm \( A \) is said to be \emph{$\gamma$-acceptably fair} for $\gamma \ge 1$ if, for any input sequence \( \sigma \) we have:
\[
\max_{i \in [k]} c_i \leq  \gamma \cdot \OPT(\sigma),
\]
where \( \OPT(\sigma) \) is the offline optimal cost for input \( \sigma \).
\end{definition}

We now explore the relationships between different notions of fairness. In particular, we show that under mild assumptions on competitiveness, certain fairness definitions imply others. These reductions help unify the space of fairness guarantees and show that designing algorithms with one fairness objective can lead to others as a consequence.

We should mention that these notions are closely related to $(\alpha,\beta)$-fairness, but not equivalent. 
In particular, any $\beta$-additively fair algorithm is also 
$(1,\Theta(\beta))$-fair under our $(\alpha,\beta)$-fairness definition, 
since bounded pairwise differences imply that no server deviates from the 
average cost by more than an additive $\Theta(\beta)$ term. 
The converse, however, does not hold: an algorithm may be 
$(1,\beta)$-fair without satisfying additive fairness, as the latter 
requires uniform bounds on all pairwise cost differences.

Similarly, any $\alpha$-multiplicatively fair algorithm is 
$(\Theta(\alpha),0)$-fair, since bounding the ratio between server costs 
implies a proportional bound relative to the total cost. 
Finally, $\gamma$-acceptable fairness directly corresponds to 
$(\gamma,0)$-fairness with respect to $\OPT(\sigma)$, since it bounds the 
maximum cost incurred by any server by a constant multiple of the offline 
optimum.

\subsection{Equivalence of Multiplicative and Acceptable Fairness for Competitive Algorithms}

We begin by showing that any algorithm that is multiplicatively fair and \( O(k) \)-competitive is also acceptably fair.

\begin{theorem}
\label{thm:mult-to-acceptable}
Any online algorithm that is $O(1)$-multiplicatively fair and \( O(k) \)-competitive is $O(1)$-acceptably fair.
\end{theorem}

\begin{proof}
By \( \alpha \)-multiplicative fairness, we have \( c_j \leq \alpha \cdot c_i \) for all \( i, j \in [k] \), so:
\[
\sum_{i=1}^k c_i \geq k \cdot \frac{1}{\alpha} \cdot c_j.
\]
By \( O(k) \)-competitiveness, \( \sum_{i=1}^k c_i \leq O(k \cdot \OPT) \), so, using that $\alpha = O(1)$, for any server \( j \):
\[
c_j \leq \frac{O(k\cdot OPT)}{k\cdot \frac{1}{\alpha}} = O(\OPT).
\]
\end{proof}

Next, we show that any algorithm that is acceptably fair can be converted into a multiplicatively fair algorithm with only a constant-factor blowup in competitiveness. Combined with Theorem~\ref{thm:mult-to-acceptable}, this establishes that, for $O(k)$-competitive algorithms, the existence of multiplicatively fair and acceptably fair algorithms is essentially equivalent.

Algorithm~\ref{alg:acctomul} implements a simple load–equalization procedure that converts acceptable fairness into multiplicative fairness. The high-level idea is as follows: at every time step, we compare each server’s cumulative cost to the current maximum cumulative cost \(h_t\). If a server is “too light’’ (i.e., its cumulative cost is significantly below \(h_t\)), the algorithm forces it to perform a corrective move of length \(\Theta(\diam)\), thereby raising its cost toward the maximum. Servers that already match the maximum cost are left unchanged. Because each corrective move has length at most \(2\diam\), whenever a server falls behind it can be brought back within a controlled additive distance of \(h_t\) at a cost of only \(O(\diam)\) additional movement. Next, we present a theorem that uses Algorithm~\ref{alg:acctomul} to show how applying this transformation to an acceptably fair algorithm yields a multiplicatively fair one.

\begin{algorithm}[t]\LinesNumbered
\caption{Conversion from an Acceptably Fair Algorithm to a Multiplicatively Fair Algorithm}
\label{alg:acctomul}

\KwIn{$k$-server algorithm $A$ that moves each server at most $\diam$ at each time step.}

Define $cc_i(t,\mathrm{ALG}) = \sum_{j=1}^t c_i(j,\mathrm{ALG})$ as the cumulative cost of server $i$ in algorithm $\mathrm{ALG}$ up to step $t$.\\
$t \leftarrow 1$\;

\ForEach{request $\sigma_t$}{
    Simulate $A(\sigma_t)$     \tcp*{run $\sigma_t$ on Algorithm $A$}
    Move all servers according to the simulation of Algorithm $A$\;
    $h_t \leftarrow \max_i \, cc_i(t, A)$\;
    \For{$i = 1$ \KwTo $k$}{
        \If{$cc_i(t,\text{this}) < h_t$}{ \label{line:check}
            Pick $x$ such that $d(s_i(t), x) \geq \diam/2$\; \label{line:find}
            Move server $i$ to $x$ and back to $s_i(t)$\; \label{line:step-end}
        }
    }
    $t \leftarrow t + 1$\;
}

\For{$i = 1$ \KwTo $k$}{
    \If{$cc_i(t,\text{this}) = h_t$}{
        Pick $x$ such that $d(s_i(t), x) \geq \diam/2$\;
        Move server $i$ to $x$ and back to $s_i(t)$\;
    }
}
\end{algorithm}

\begin{theorem}
\label{thm:atom}
Any online algorithm that is $O(1)$-acceptably fair (and so \( O(k) \)-competitive as a consequence) can be transformed into an $O(1)$-multiplicatively fair algorithm with a competitive ratio $O(k)$.
\end{theorem}

\begin{proof}
Define \( cc_i(t,\mathrm{ALG}) = \sum_{j=1}^t c_i(j,\mathrm{ALG}) \) as the cumulative cost of server \( i \) up to request \( t \) under algorithm \( \mathrm{ALG} \).

First, to justify line~\ref{line:find}, observe that in any metric space with diameter \(\diam\), for any point \(x\), there exists a point \(y\) such that \(d(x,y) \ge \diam/2\). Otherwise, all points would lie within a ball of radius less than \(\diam/2\), implying that the maximum distance between any two points is strictly less than \(\diam\), which contradicts the definition of diameter. Consequently, if a server moves from its current position to such a point and returns, it travels a distance of at least \(\diam\) and at most \(2\diam\).


Second, we show that all servers pay a cost of at least \(h_t = \max_j cc_j(t, A)\) at the end of each step (Line~\ref{line:step-end}).

\begin{lemma}
\label{lem:lx}
For all servers \( i \) and all time steps \( t \), we have
\[
cc_i(t,\text{this}) \;\ge\; h_t \;=\; \max_j cc_j(t, A)
\]
after Line~\ref{line:step-end} in Algorithm~\ref{alg:acctomul}.
\end{lemma}
\begin{proof}
We prove by induction that after the end of each step \( t \) (Line~\ref{line:step-end}), it holds that \( cc_i(t,\text{this}) \ge h_t \) for all \( i \).  
For the base case \( t = 0 \), we have \( cc_i(0,\text{this}) = h_0 = 0 \), so the statement holds.  

Now consider any step \( t > 0 \). By the definition of Algorithm~$A$ (see Algorithm~\ref{alg:acctomul}), the maximum cumulative cost in the simulated execution of \( A \) increases by at most \( \diam \) between two consecutive steps; that is,
\[
h_t \le h_{t-1} + \diam.
\]
If, at Line~\ref{line:check}, we have \( cc_i(t-1,\text{this}) < h_t \), then the algorithm performs a corrective move that adds at least \( \diam \) to server \( i \). Combining this with the inductive hypothesis and the bound on the increase of \( h_t \), we obtain
\[
cc_i(t,\text{this})
\;\ge\;
cc_i(t-1,\text{this}) + \diam
\;\ge\;
h_{t-1} + \diam
\;\ge\;
h_t.
\]
If instead \( cc_i(t-1,\text{this}) \ge h_t \), then no corrective move is made, and hence
\[
cc_i(t,\text{this}) = cc_i(t-1,\text{this}) \ge h_t.
\]
Thus, in both cases the inequality \( cc_i(t,\text{this}) \ge h_t \) holds at the end of step \( t \), completing the induction.
\end{proof}

Next, we complete the bound on the cost paid by each server, which will allow us to show that Algorithm~\ref{alg:acctomul} is multiplicatively fair.

\begin{lemma}
\label{lem:l16}
The total cost paid by any server in Algorithm~\ref{alg:acctomul} lies in the interval
\[
\bigl[\,\max(\max_j cc_j(A),\,\diam),\, 2 \cdot \max_i c_i + 2 \cdot \diam\,\bigr].
\]
\end{lemma}

\begin{proof}
Fix a server \( i \). At each time step, server \( i \) may move either due to simulating algorithm \( A \), or because its cumulative cost \( cc_i(t,\text{this}) \) is less than the current maximum \( h_t \). Let \( t^* \) be the last step for which \( cc_i(t^*,\text{this}) < h_{t^*} \) in Line~\ref{line:check} of Algorithm~\ref{alg:acctomul}, meaning that the server performs a corrective move (note that \( t^* \) could be \( 0 \)).

After time \( t^* \), server \( i \) moves only according to algorithm \( A \), and thus its cost from that point onward is at most \( cc_i(A) \). At time \( t^* \), server \( i \)'s cost up to step \( t^* \) was strictly less than \( h_{t^*} \), and since \( h_{t^*} \le \max_j cc_j(T,A) \), this gives an upper bound on the cost before the corrective move. Each corrective move (see Line~\ref{line:find}) also adds at most \( 2\diam \), so in total, server \( i \)'s movement is bounded by \( \max_j cc_j(A) + 2\diam \).

It is worth mentioning that the last four lines of Algorithm~\ref{alg:acctomul} only move the servers whose cost satisfies
\[
cc_i(T,\text{this}) = h_T = \max_j cc_j(A),
\]
and these moves force their final costs to lie in the range
\[
[\max_j cc_j(A) + \diam,\; \max_j cc_j(A) + 2\diam].
\]
This makes sure that our upper bound is always correct for all servers, and for these servers, our lower bound also holds.

For all other servers,by Lemma~\ref{lem:lx}, they pay at least \( \max_j cc_j(A) \), and any server that pays more(not equal) must have been moved in Line~\ref{line:find}, which contributes at least an additional \( \diam \). Therefore, fortm this fact and Lemma~\ref{lem:lx} every server pays at least
\[
\max(\diam,\; \max_j cc_j(A)),
\]
and at most
\[
2\max_j cc_j(A) + 2\diam.
\]

\end{proof}

Lemma~\ref{lem:l16} ensures that, after transforming \( A \) with Algorithm~\ref{alg:acctomul}, every server incurs a cost of at least \( \max(\max_j c_j,\, \diam) \) and at most \( 2 \cdot \max_j c_j + 2 \cdot \diam \), ensuring that the resulting algorithm is multiplicatively fair with ratio at most \( 4 \le O(1) \).

Next, since the algorithm was \( O(1) \)-acceptably fair, computing the total cost gives
\[
\sum_{i=1}^k c_i 
\;\le\; k \cdot 2 \cdot \max_j c_j \;+\; 2k \cdot \diam
\;\le\; O(k) \cdot \OPT,
\]
which shows that the transformed algorithm is also \( O(k) \)-competitive.

\end{proof}


\subsection{Connections to Additive Fairness}

While we have shown that acceptable and multiplicative fairness are closely related under reasonable competitiveness assumptions 
, additive fairness appears to be fundamentally harder to guarantee in the online setting. In particular, additive fairness requires tight control over cost deviations between servers—something that may be impossible to achieve without knowing when the input sequence ends. 
One might think that using Algorithm~\ref{alg:acctomul} we could transform an acceptably fair algorithm into an additively fair one. However, Algorithm~\ref{alg:acctomul} only keeps servers within a constant \emph{multiplicative} factor of each other (and, in fact, this factor can be improved to $2$), whereas achieving additive fairness would require that at all times every server stays within distance \(\diam\) of the current maximum-cost server. This is problematic in an online setting: consider an adversarial sequence in which, in turn, server~1, then server~2, then server~3, and so on, becomes the unique maximum-cost server. Each time the identity of the maximum changes, an additive-fairness transformation must force the remaining \(k-1\) servers to ``catch up'' by moving them an extra distance \(\Theta(\diam)\), incurring about \((k-1)\diam\) additional cost per phase beyond what the original algorithm $A$ pays. Repeating this pattern over many phases yields a total overhead of order \(k\diam\) per phase, leading to a multiplicative blow-up of roughly \(k\) in the competitive ratio. Thus, any strategy that tries to maintain strict \(\diam\)-additive fairness at all times appears incompatible with preserving the original \(O(k)\)-competitive guarantee.


When the algorithm knows when the request sequence ends, additive fairness becomes achievable; however, this guarantee is weak, since the algorithm is only additively fair at the \emph{final} step rather than throughout the entire execution. Under this additional assumption, we can nevertheless construct an additively fair algorithm from an acceptably fair one, as shown next.

\begin{theorem}
\label{thm:lastAA}
If the algorithm knows the final request in advance, then any $O(1)$-acceptably fair algorithm can be transformed into a $O(\diam)$-additive fair algorithm with competitive ratio $O(k)$.
\end{theorem}



\begin{proof}
Let \( A \) be any acceptably fair algorithm, and let \( c_i \) denote the final cost of server \( i \) under \( A \), with \( c_{\max} := \max_i c_i \). Since \( A \) is acceptably fair, \( c_{\max} = O(\OPT) \).

After simulating \( A \), for each server \( i \) such that \( c_i < c_{\max} - \diam \), we increase its cost by performing one or more additional moves to points at distances in the range \( [\diam/2, \diam] \) (This point exists; see the first paragraph of proof of Theorem~\ref{thm:atom}). Each such move increases \( c_i \) by at least \( \diam/2 \) and at most \( \diam \). As a result, the final cost of every server lies in the interval \( [c_{\max} - \diam, c_{\max}] \), since we cannot exceed this range with steps of length at most \( \diam \). This ensures additive fairness with slack \( \diam \).

Finally, since \( c_{\max} = O(\OPT) \), and all servers lie within \( \diam \) of it, the total cost is at most
\[
\sum_{i=1}^k c_i \leq k \cdot c_{\max} = O(k \cdot \OPT),
\]
preserving \( O(k) \)-competitiveness.
\end{proof}

This motivates the following conjecture, which captures a potential separation between acceptable and additive fairness.

\begin{conjecture}
\label{conj:additive-barrier}
In the absence of knowledge of the final request, it is not possible to transform every $O(1)$-acceptably fair online algorithm into an $O(\diam)$-additive fair algorithm with a competitive ratio of $O(k)$.
\end{conjecture}

Finally, we show that additive fairness can be achieved, but only at the cost of a significant deterioration in competitive ratio. In particular, the construction below should be viewed merely as a baseline: the resulting \(c \cdot k\) competitive ratio is quite weak in general and substantially worse than what is suspected to be achievable through a transformation from acceptable or multiplicative fairness.

\begin{theorem}
\label{thm:kkcomp}
Let $A$ be a $c$-competitive online algorithm for the $k$-server problem. Then $A$ can be transformed into a $\diam$-additively fair algorithm with competitive ratio $c\cdot k$.
\end{theorem}

\begin{proof}
Assume we apply a similar idea from Theorem~\ref{thm:lastAA}. 
We run the original algorithm \(A\), and each time a server gets more than \(2\diam\) behind the maximum server cost, we perform a correction move as many times as we want to bring it within \(2\diam\) of the maximum. This action adds between \(\diam\) and \(2\diam\), so the server will not become the maximum by performing this action(since we stop if it gets into the range of $2\diam$). This results in \(2\diam\)-additive fairness.
In the transformed algorithm, for every server \(i\), we have:
\[
cc_i(t,\text{this}) \geq \max_j cc_j(t,\text{this}) - 2\diam, 
\quad \text{and} \quad 
cc_i(t,\text{this}) \leq \max_j cc_j(t,\text{this}),
\]
ensuring additive fairness.

We now prove by induction that, at every time step \(t\), the maximum cost in the transformed algorithm is upper-bounded by the total cost of algorithm \(A\) up to time \(t\), i.e.,
\[
\max_j cc_j(t,\text{this}) \leq \text{cost}(A,t).
\]

\noindent \textbf{Base case (\(t = 0\)):}  
Both the maximum cost and \(\text{cost}(A,0)\) are zero.

\smallskip
\noindent \textbf{Inductive step:}  
Suppose the claim holds for time \(t-1\). Let \(c_i(t)\) be the cost incurred by server \(i\) in \(A\) at time \(t\).
We know that in our algorithm the maximum only changes due to the movement of the server according to \(A\), and the correction step does not change the maximum.  
So the maximum cost will be:
\begin{equation}
m = \arg\max_j cc_j(t,\text{this}), 
\qquad
cc_m(t,\text{this}) = cc_m(t-1,\text{this}) + c_m(t).
\end{equation}
Now we can say, based on the inductive hypothesis,
\begin{equation}
cc_m(t-1,\text{this}) \leq \max_j cc_j(t-1,\text{this}) \leq \text{cost}(A, t-1),
\end{equation}
which results in
\begin{equation}
cc_m(t,\text{this}) =cc_m(t-1,\text{this}) + c_m(t) 
\leq \text{cost}(A, t-1) + c_m(t) 
\leq \text{cost}(A, t).
\end{equation}
which proves the induction.

For the final time step \(t\), this gives:
\[
\max_j c_j(\text{this}) \leq \text{cost}(A).
\]

Thus, the total cost of the transformed algorithm is:
\[
\sum_{i=1}^k c_i(\text{this}) 
\leq k \cdot \max_j c_j(\text{this}) 
\leq k \cdot \text{cost}(A).
\]

Since \(A\) is \(c\)-competitive, the total cost is at most \(c \cdot \OPT\), hence:
\[
\text{cost}(\text{this}) \leq c \cdot k \cdot \OPT.
\]

This completes the proof.

\end{proof}
Finally, we note that for $k= O(1)$ we can argue that deterministic algorithms exist which are $1$-competitive and additively fair.

\begin{theorem}
For any constant $k$, there exists a deterministic algorithm for the $k$-server problem that is $O(1)$-competitive and $\diam$-additively fair.
\end{theorem}

\begin{proof}
By applying Theorem~\ref{thm:kkcomp} to any $O(k)$-competitive algorithm (e.g., DCA on trees), we obtain a $k \cdot O(k)$-competitive algorithm that is also additively fair. Since $k$ is constant, the competitive ratio remains $O(1)$.
\end{proof}

\section{Egalitarian Cost Model}
\label{app:egalitarian}

In the \emph{egalitarian} cost model for the $k$-server problem, the cost of an algorithm
is defined as the maximum distance traveled by any single server, rather than the sum
of distances traveled by all servers.
For a request sequence $\sigma$, let $\mathrm{OPT}_e(\sigma)$ denote the minimum possible
egalitarian cost achievable by any offline solution.

Considering the optimal egalitarian solution total cost of all servers is less than k times the maximum server cost so we can write
\[
\mathrm{OPT}_e(\sigma)\times k  \;\ge\;cost(\mathrm{OPT}_e(\sigma),\sigma)   \;\ge\; {\mathrm{OPT}(\sigma)}.
\]

Thus, the optimal egalitarian cost is at least the average optimal server cost of the normal model,
\[
\mathrm{OPT}_e(\sigma) \;\ge\; \frac{\mathrm{OPT}(\sigma)}{k},
\]
where $\mathrm{OPT}(\sigma)$ denotes the optimal utilitarian (total) cost.

Now suppose that an algorithm $A$ is $(\alpha,\beta)$-fair with respect to the utilitarian
cost, meaning that for every server $i$,
\[
c_i(A,\sigma) \;\le\; \frac{\alpha}{k}\,\mathrm{cost}(A,\sigma) + \beta.
\]
Then the egalitarian cost of $A$ satisfies
\[
A_e(\sigma)
\;=\;
\max_{i \in [k]} c_i(A,\sigma)
\;\le\;
\frac{\alpha}{k}\,\mathrm{cost}(A,\sigma) + \beta.
\]

If, in addition, $A$ is $c$-competitive under the utilitarian objective, i.e.,
$\mathrm{cost}(A,\sigma) \le c \cdot \mathrm{OPT}(\sigma)$ for all $\sigma$, then
\[
A_e(\sigma)
\;\le\;
\alpha c \cdot \mathrm{OPT}_e(\sigma) + \beta,
\]
where the inequality follows from $\mathrm{OPT}_e(\sigma) \ge \mathrm{OPT}(\sigma)/k$.
Thus, $(\alpha,\beta)$-fairness under the utilitarian cost directly implies an $\alpha c$ competitive
guarantee in the egalitarian cost model, up to the same multiplicative factor and additive
slack.

\section{Omitted Proofs}

We now give full proofs for our main results taht were omitted from the main body of the paper.
\label{sec:proofs}

\subsection{Proof of Theorem~\ref{thm:fairness}}
\begin{theorem}[Restatement of Theorem~\ref{thm:fairness}]
Consider any $k$-server input sequence 
$\sigma$ of length $T$ over a metric space of diameter $\diam$. Let $\{{s}_i(t)\}_{i \in [k], t \in [T]}$ be a sequence of server positions that achieves optimal cost on this input sequence. 

For any $\varepsilon > 0$, Algorithm~\ref{alg:fair_cost} takes $\{{s}_i(t)\}_{i \in [k], t \in [T]}$ as input and returns a modified sequence of positions $\{\hat {s}_i(t)\}_{i \in [k], t \in [T]}$ that is $( 1 + \varepsilon, \beta)$-fair with respect to $\OPT(\sigma)$, for 
$
\beta = O\left( \frac{\log k \cdot \diam}{\varepsilon} \right)$. 
Moreover, the total cost of the resulting solution is bounded by $\OPT(\sigma) +  O \left (\frac{k \log k \cdot \diam}{\varepsilon}\right ).$
\end{theorem}

\begin{proof}
\label{apx:fairness}
We use the term \textbf{swapping} to refer to the act of exchanging the future work plans of two servers starting from a designated time step. Specifically, if servers $H_r$ and $L_r$ swap at time $z$, each server takes over the remaining assignments of the other from time $z$ onward(Line ~\ref{line:swapc}). 

Swapping increases the cost of each involved server by, at most, the distance between their positions at the swap time. This is reflected in the transition step at time $t = z$, where each server incurs an additional penalty equal to $d(s_{H_r}^{(r)}(z), s_{L_r}^{(r)}(z))$. Since the distance between any two points is at most $\diam$, the total cost increase per swap is at most $2 \cdot \diam$.

Let $r$ denote the number of swaps performed by the algorithm so far. Define $\hat c_x^{(r)}$ to be the total cost incurred by server $x$ after $r$ swaps(which includes swap costs too).
Let $H_r, L_r \in [k]$ denote the servers with the maximum and minimum total costs respectively before the $r$-th round of swapping (see Line~\ref{line:mnmx} in Algorithm~\ref{alg:fair_cost})

We use $\hat c_{H_r}^{(r)}$ to denote the total cost incurred by server $H_r$ at round $r$, and $\hat c_{L_r}^{(r)}$ for the cost of server $L_r$. Our analysis focuses on bounding the difference $\hat c_{H_r}^{(r)} - \hat c_{L_r}^{(r)}$ and showing that it decreases by a constant factor in each round.

To formalize the cost-balancing process between maximum and mimum cost servers, we focus on the problem of comparing two sequences of cost values
Each server's future work plan can be viewed as a sequence of per-request costs. When we consider swapping the remaining workload of two servers, we want to find a split point $z$ such that exchanging their suffixes beyond $z$ balances their total costs up to a small additive slack. The following lemma guarantees the existence of such a split point.

\begin{lemma}
\label{lem:l1}
Let $A = (A_1, A_2, \dots, A_n)$ and $B = (B_1, B_2, \dots, B_n)$ be two sequences of real numbers, where $A_i, B_i \in [0, c]$ for all $i \in [n]$. Then there exists an index $z \in [n]$ such that:
\begin{equation}
\label{eq:lem5}
    \left|  \left( \sum_{i=1}^{z} A_i + \sum_{i=z+1}^{n} B_i  \right)- \left( \sum_{i=1}^{z} B_i + \sum_{i=z+1}^{n} A_i \right) \right| \leq c.
\end{equation}
\end{lemma}

\begin{proof}
\label{app:lemma_l1}
Define a sequence $D_j$ for $j = 0$ to $n$ as:
\[
D_j = \sum_{i=1}^{j} A_i + \sum_{i=j+1}^{n} B_i -  \sum_{i=1}^{j} B_i - \sum_{i=j+1}^{n} A_i .
\]
By defintion of $D$ it suffices to show that there exists $z$ such that $|D_z|\leq c$. 
Now observe that:
\[
D_0 = \sum_{i=1}^{n} B_i - \sum_{i=1}^{n} A_i, \quad D_n = \sum_{i=1}^{n} A_i - \sum_{i=1}^{n} B_i = -D_0.
\]
So $D_0$ and $D_n$ are equal in magnitude but have opposite signs.
We now show that there exists an index $x \in \{0, \dots, n-1\}$ such that $D_x$ and $D_{x+1}$ have opposite signs (or one is zero). Suppose for contradiction that no such index exists—that is, for all $j = 0, \dots, n-1$, $D_j$ and $D_{j+1}$ have the same sign, and none are zero.

In that case, the entire sequence $(D_0, D_1, \dots, D_n)$ has the same sign throughout. But this contradicts the fact that $D_0$ and $D_n$ have opposite signs. 
Therefore, such an index $x$ must exist.
Now consider the change between consecutive terms: 
 \begin{align*}
        D_{x}-D_{x-1} &=  \sum_{i=1}^{x} A_i + \sum_{i=x+1}^{n} B_i - \sum_{i=1}^{x} B_i - \sum_{i=x+1}^{n} A_i - \sum_{i=1}^{x-1} A_i -\sum_{i=x}^{n} B_i + \sum_{i=1}^{x-1} B_i + \sum_{i=x}^{n} A_i \\
D_x - D_{x-1} &= A_x - B_x - B_x + A_x = 2(A_x - B_x), \\
\Rightarrow |D_x - D_{x-1}| &= 2|A_x - B_x| \leq 2c.
\end{align*}
since $A_{x}, B_{x} \in [0, c]$.

We now argue that either $x$ or $x+1$ satisfies the lemma. There are two cases:
\begin{itemize}
    \item If $|D_x| \leq c$, then $x$ is a valid choice for $z$.
    \item Otherwise, since $D_x$ and $D_{x+1}$ have opposite signs and differ by at most $2c$ and $|D_x|>c$, it must be that $|D_{x+1}| < c$, making $x+1$ a valid choice.
\end{itemize}
Hence, there exists some $z \in \{x, x+1\}$ such that:
\[
\left| \sum_{i=1}^{z} A_i + \sum_{i=z+1}^{n} B_i - \left( \sum_{i=1}^{z} B_i + \sum_{i=z+1}^{n} A_i \right) \right| \leq c,
\]
as claimed.
\end{proof}

For each swap between $H_r$ and $L_r$, consider their per-request costs at each step $t$, denoted $\{\hat c_{H_r}^{(r)}(t)\}$. This sequence has length $n$, and each entry is at most $\diam$. By applying Lemma~\ref{lem:l1}, we can find an index $z$ such that $|D_z| \leq \diam$ when swapping the suffixes of these sequences.

To proceed with the proof, we establish three key claims.

\begin{claim}
\label{claim1}
In Algorithm~\ref{alg:fair_cost} with $\beta=2 \cdot 1+\varepsilon \cdot diam \cdot (\frac{3}{2}+\log_{\frac{2+2\varepsilon}{2+\varepsilon}} k)$, after swapping servers $H_{r-1}$ and $L_{r-1}$ in round $r\leq k\cdot \log_{\frac{2+2\varepsilon}{2+\varepsilon}} k$, the total costs of both servers are reduced to at most a $\left( \frac{2+\varepsilon}{2+2\varepsilon}\right)$ fraction of the maximum server cost before the swap. That is,
\begin{align*}
\hat c_{H_{r-1}}^{(r)} &\leq \left( \frac{2+\varepsilon}{2+2\varepsilon} \right) \cdot \hat c_{H_{r-1}}^{(r-1)},  \\
\hat c_{L_{r-1}}^{(r)} &\leq \left( \frac{2+\varepsilon}{2+2\varepsilon} \right) \cdot \hat c_{H_{r-1}}^{(r-1)}. 
\end{align*}
\end{claim}

After each swap, the two involved servers exchange their remaining future responsibilities, which balances their total cost up to an additive penalty of $\diam$. Using Lemma~\ref{lem:l1}, we show that this exchange brings each server’s cost close to the average of their prior costs. By bounding the additive terms using the swap threshold $\beta$ and the maximum number of rounds, we show that both servers' costs shrink by a constant multiplicative factor.

\begin{proof}
\label{app:claim_1}
After swapping in round $r$, server $H_{r-1}$ is reassigned to perform the remaining tasks originally assigned to $L_{r-1}$, and vice versa.

Let $A_i$ denote the cost incurred by server $H_{r-1}$ at time $i$, and let $B_i$ denote the cost incurred by server $L_{r-1}$ at time $i$. We have:
\[
\hat c_{H_{r-1}}^{(r-1)} = \sum_{i=1}^{T} A_i, \quad \text{and} \quad \hat c_{L_{r-1}}^{(r-1)} = \sum_{i=1}^{T} B_i.
\]

Suppose we apply the swap at the index $z \in [T]$ guaranteed by Lemma~\ref{lem:l1}. Then, the new cost of server $H_{r-1}$ in round $r$ is:
\begin{equation}
\label{eq:b1}
\hat c_{H_{r-1}}^{(r)} \leq \sum_{i=1}^{z-1} A_i + \diam + \sum_{i=z}^{T} B_i,
\end{equation}
and the new cost of server $L_{r-1}$ is:
\begin{equation}
\label{eq:b2}
\hat c_{L_{r-1}}^{(r)} \leq \sum_{i=1}^{z-1} B_i + \diam + \sum_{i=z}^{T} A_i.
\end{equation}

We are now aiming to show that, after swapping, the cost incurred by each server is close to the average of their original costs, up to an additive penalty. 
Let:
\begin{align*}
F &= \frac{1}{2} \left( \sum_{i=1}^{z-1} A_i + \sum_{i=z}^{T} B_i + \sum_{i=1}^{z-1} B_i + \sum_{i=z}^{T} A_i \right), \\
G &= \frac{1}{2} \left( \sum_{i=1}^{z-1} A_i + \sum_{i=z}^{T} B_i - \sum_{i=1}^{z-1} B_i - \sum_{i=z}^{T} A_i \right).
\end{align*}
Then, applying the triangle inequality to bound $
\left| F + G \right| \leq \left| F \right| + \left| G \right|$ gives:
\[
\sum_{i=1}^{z-1} A_i + \sum_{i=z}^{T} B_i = F + G 
\leq \frac{1}{2} \left( \sum_{i=1}^{T} A_i + \sum_{i=1}^{T} B_i \right)
+ \frac{1}{2} \left| \sum_{i=1}^{z-1} A_i + \sum_{i=z}^{T} B_i - \sum_{i=1}^{z-1} B_i - \sum_{i=z}^{T} A_i \right|.
\]
By Lemma~\ref{lem:l1}, the absolute value term is at most $\diam$. Therefore:
\[
\sum_{i=1}^{z-1} A_i + \sum_{i=z}^{T} B_i 
\leq \frac{1}{2} \left( \sum_{i=1}^{T} A_i + \sum_{i=1}^{T} B_i \right)
+ \frac{\diam}{2}.
\]
Since $\hat c_{H_{r-1}}^{(r-1)} = \sum_{i=1}^{T} A_i$ and $\hat c_{L_{r-1}}^{(r-1)} = \sum_{i=1}^{T} B_i$, we obtain:
\[
\sum_{i=1}^{z-1} A_i + \sum_{i=z}^{T} B_i 
\leq \frac{1}{2} \hat c_{H_{r-1}}^{(r-1)} + \frac{1}{2} \hat c_{L_{r-1}}^{(r-1)} + \frac{\diam}{2}.
\]

A symmetric argument applies to the cost of server $L_{r-1}$ after the swap. Specifically, by applying the same triangle inequality reasoning but using $-G$ instead of $G$, we obtain:
\begin{equation}
    \label{eq:b3}
\sum_{i=1}^{z-1} B_i + \sum_{i=z}^{T} A_i 
\leq \frac{1}{2} \hat c_{H_{r-1}}^{(r-1)} + \frac{1}{2} \hat c_{L_{r-1}}^{(r-1)} + \frac{\diam}{2}.
\end{equation}
Combining Equation~\eqref{eq:b3} with Equations~\eqref{eq:b1} and~\eqref{eq:b2}, we conclude:
\begin{align}
\label{eq:ee1}
\hat c_{H_{r-1}}^{(r)} &\leq \frac{3}{2} \cdot \diam + \frac{1}{2} \hat c_{H_{r-1}}^{(r-1)} + \frac{1}{2} \hat c_{L_{r-1}}^{(r-1)}, \\
\hat c_{L_{r-1}}^{(r)} &\leq \frac{3}{2} \cdot \diam + \frac{1}{2} \hat c_{H_{r-1}}^{(r-1)} + \frac{1}{2} \hat c_{L_{r-1}}^{(r-1)} \notag
\end{align}


Let $\OPT$ denote the total cost of the initial (optimal) offline solution. Since each swap in Algorithm~\ref{alg:fair_cost} adds at most $2 \cdot \diam$ to the total cost (one transition penalty per server), the total cost at the beginning of round $r$ is at most:
\[
\text{TotalCost}^{(r-1)} \leq \OPT + 2(r-1) \cdot \diam.
\]
Hence, the average cost per server before round $r$ is:
\[
\frac{1}{k} \cdot \left( \OPT + 2(r-1) \cdot \diam \right).
\]
Since $L_{r-1}$ is the server with the \textbf{minimum} cost at the start of round $r$, it must be that:
\begin{equation}
\label{eq:a20}
\hat c_{L_{r-1}}^{(r-1)} \leq \frac{1}{k} \left( \OPT + 2(r - 1) \cdot \diam \right).
\end{equation}

Since the algorithm proceeds to perform a swap in round $r$, it must have been the case that the fairness condition in Algorithm~\ref{alg:fair_cost} was violated. In particular, the server $H_{r-1}$ with the maximum cost must satisfy:
\begin{equation}
\label{eq:a21}
\hat c_{H_{r-1}}^{(r-1)} > \frac{(1+\varepsilon) \cdot \OPT}{k} + \beta.
\end{equation}
Using inequality~\eqref{eq:a20} for $\hat c_{L_{r-1}}^{(r-1)}$ and~\eqref{eq:a21} for $\hat c_{H_{r-1}}^{(r-1)}$, we combine the two to get:
\begin{equation}
  \label{eq:aa22}  
\hat c_{L_{r-1}}^{(r-1)} 
\leq \frac{1}{k} \left( 2(r - 1) \cdot \diam \right) 
+ \frac{1}{1+\varepsilon} \cdot \hat c_{H_{r-1}}^{(r-1)} 
- \frac{\beta}{1+\varepsilon}.
\end{equation}

Substituting Equation~\eqref{eq:aa22} into Equation~\eqref{eq:ee1} we have:
\begin{align}
\hat c_{H_{r-1}}^{(r)} 
&\leq \frac{3}{2} \cdot \diam 
+ \frac{1}{2} \hat c_{H_{r-1}}^{(r-1)} 
+ \frac{1}{2} \left( \frac{\hat c_{H_{r-1}}^{(r-1)}}{1+\varepsilon} - \frac{\beta}{1+\varepsilon} + \frac{2(r - 1) \cdot \diam}{k} \right) \label{eq:c} \\
&= \left( \frac{1}{2} + \frac{1}{2(1+\varepsilon)} \right) \hat c_{H_{r-1}}^{(r-1)} - \frac{\beta}{2(1+\varepsilon)} + \frac{3}{2} \cdot \diam + \frac{(r - 1) \cdot \diam}{k}. \nonumber
\end{align}

To conclude the proof, we substitute the chosen value of $\beta$ into~\eqref{eq:c}. Since the algorithm performs at most $r \leq k \cdot \log_{\frac{2+2\varepsilon}{2+\varepsilon}} k$ rounds, the additive terms in~\eqref{eq:c} are fully absorbed by setting:
\begin{align*}
\frac{\beta}{2+2\varepsilon} &= \frac{( 2 \cdot (1+\varepsilon) \cdot \diam \cdot \left( \frac{3}{2} + \log_{\frac{2+2\varepsilon}{2+\varepsilon}} k \right)}{2(1+\varepsilon)} = \frac{3}{2}\cdot\diam+\diam\cdot \log_{\frac{2+2\varepsilon}{2+\varepsilon}} k,\\
&\geq \frac{3}{2}\cdot\diam+\diam\cdot \frac{r}{k}\geq \frac{3}{2} \cdot \diam + \frac{(r - 1) \cdot \diam}{k}.
\end{align*}
This reduces inequality~\eqref{eq:c} to a purely multiplicative bound:
\[
\hat c_{H_{r-1}}^{(r)} \leq \left( \frac{1}{2} + \frac{1}{2(1+\varepsilon)} \right) \hat c_{H_{r-1}}^{(r-1)} = \left( \frac{2+\varepsilon}{2+2\varepsilon} \right) \hat c_{H_{r-1}}^{(r-1)}.
\]
A symmetric argument works for $\hat c_{L_{r-1}}^r$, satisfying the claim.

\end{proof}

The next claim shows that the maximum cost among all servers is decreasing across rounds, and only the swapped servers may change their cost.

\begin{claim}
\label{claim2}
Assume that $r \leq k \cdot \log_{\frac{2+2\varepsilon}{2+\varepsilon}} k$. Then:
\begin{align}
\hat c_{H_r}^{(r)} &\leq \hat c_{H_{r-1}}^{(r-1)}, \label{eq:decreasing} \\
\text{and for all } x \notin \{H_{r-1}, L_{r-1}\}, \quad \hat c_x^{(r)} &= \hat c_x^{(r-1)}. \label{eq:same}
\end{align}
\end{claim}
\begin{proof}
From Claim~\ref{claim1}, the costs of both $H_{r-1}$ and $L_{r-1}$ decrease by at least a multiplicative factor of $\left( \frac{2+\varepsilon}{2+2\varepsilon} \right) < 1$. Since $H_{r-1}$ had the maximum cost before the swap, this implies that the new maximum cost after round $r$ must be less than or equal to the previous maximum:
\[
\hat c_{H_r}^{(r)} \leq \hat c_{H_{r-1}}^{(r-1)}.
\]

Moreover, all other servers $x \notin \{H_{r-1}, L_{r-1}\}$ are not affected by the swap, and their cost assignments are carried over unchanged in the algorithm (as seen in the cost copying step before the swap in Line~\ref{line:copy} of Algorithm~\ref{alg:fair_cost}). Hence,
\[
\hat c_x^{(r)} = \hat c_x^{(r-1)} \quad \text{for all } x \notin \{H_{r-1}, L_{r-1}\}.
\]\end{proof}

The following claim shows that within any window of $(k+1)$ rounds of Algorithm~\ref{alg:fair_cost}, the maximum server cost shrinks by a multiplicative factor of at least $\left( \frac{2+\varepsilon}{2+2\varepsilon} \right)$. This forms the basis for the geometric convergence of the algorithm.

\begin{claim}
\label{claim3}
Let $r$ be any round in Algorithm~\ref{alg:fair_cost} such that $r\leq k\cdot\log_{\frac{2+2\varepsilon}{2+\varepsilon}} k - k$, and suppose the algorithm runs for $k$ additional rounds beyond $r$. Then:
\begin{equation}
\hat c_{H_{r + k}}^{(r + k)} \leq \left( \frac{2+\varepsilon}{2+2\varepsilon} \right) \cdot \hat c_{H_r}^{(r)}.
\label{eq:claim3}
\end{equation}
\end{claim}


\begin{proof}[Proof Sketch]
    Over any $k+1$ rounds, at least one server must be chosen as the highest-cost server in two rounds. We analyze two cases based on whether this server was swapped as a low-cost server in between. In both cases, applying the previous claims shows that the maximum cost decreases by a constant multiplicative factor. The full argument appears in Appendix~\ref{app:claim3}.
\end{proof}

We are now ready to complete the proof of Theorem \ref{thm:fairness}. The claims established above guarantee that:
\begin{itemize}
    \item The cost of the maximum server always decreases (Claim~\ref{claim2}).
    \item The cost of any server not involved in a swap remains unchanged (Claim~\ref{claim2}).
    \item The cost of any server that is selected as the maximum after $k$ round must decrease by a multiplicative factor of at least $\left( \frac{2+\varepsilon}{2+2\varepsilon} \right)$ (Claim~\ref{claim3}).
\end{itemize}

As a result, the cost of the maximum-cost server decreases geometrically over time until it satisfies the $(1 +\varepsilon, \beta)$-fairness condition given in Equation~\eqref{eq:fairness_general}.

In the remainder of the proof, we upper bound the total number of swaps required before all server costs satisfy the fairness criterion.


From Claim~\ref{claim3}, we know that for every $k$ rounds, the maximum server cost decreases by a factor of $\left( \frac{2+\varepsilon}{2+2\varepsilon}  \right)$. Starting at round $0$, we have:


\begin{align*}
\hat c_{H_{k}}^{(k)} &\leq \left( \frac{2+\varepsilon}{2+2\varepsilon}  \right) \hat c_{H_0}^{(0)}, \label{eq:ineq-base}
\end{align*}

\begin{align*}
\hat c_{H_{2k}}^{(2k)} &\leq \left( \frac{2+\varepsilon}{2+2\varepsilon}  \right) \hat c_{H_{k}}^{(k)} 
\leq \left( \frac{2+\varepsilon}{2+2\varepsilon}  \right)^2 \hat c_{H_0}^{(0)}, \\
&\vdots \\
\hat c_{H_{qk}}^{(qk)} &\leq \left( \frac{2+\varepsilon}{2+2\varepsilon}  \right)^q \hat c_{H_0}^{(0)}.
\end{align*}
It is worth mentioning that $\hat c_{H_0}^{(0)}$ denotes the maximum cost among all servers after executing the initial offline optimal algorithm $A$, before any swaps. Since any server in algorithm $A$ incurs total cost at most $\OPT$, we have:

\begin{align}
\hat c_{H_{qk}}^{(qk)} &\leq \left( \frac{2+\varepsilon}{2+2\varepsilon}  \right)^q \hat c_{H_0}^{(0)} \nonumber \\
&\leq \left( \frac{2+\varepsilon}{2+2\varepsilon}  \right)^q \cdot \OPT. \nonumber
\end{align}
To ensure this final cost is at most $\frac{(1+\varepsilon) \cdot \OPT}{k} + \beta$, it suffices to run the algorithm $qk$ times such that:
\[
\left( \frac{2+\varepsilon}{2+2\varepsilon}  \right)^q \cdot \OPT \leq \frac{(1+\varepsilon) \cdot \OPT}{k}.
\]
Solving for $q$ gives: 
\[
q \geq \log_{\frac{2+2\varepsilon}{2+\varepsilon}} \left( \frac{k}{1+\varepsilon} \right) = \log_{\frac{2+2\varepsilon}{2+\varepsilon}} k - \log_{\frac{2+2\varepsilon}{2+\varepsilon}} (1+\varepsilon) .
\]
by setting total number of swaps to $qk = k \log_{\frac{2+2\varepsilon}{2+\varepsilon}} k$ we satisfy both fairness guarntee and Claims'~\ref{claim1},~\ref{claim2} and \ref{claim3} conditions on $r$. 
For any $\varepsilon > 0$, we have $\log_{\frac{2+2\varepsilon}{2+\varepsilon}} k = \frac{\log k}{\log \frac{2+2\varepsilon}{2+\varepsilon}} =  \Theta\left( \frac{1}{\varepsilon} \log k \right)$ so the total number of swaps $
O\left( \frac{1}{\varepsilon} \cdot k \cdot \log k \right).$
This completes the proof: after at most $O\left( \frac{1}{\varepsilon} \cdot k \cdot \log k \right)$ swaps, Algorithm~\ref{alg:fair_cost} returns a $(1+\varepsilon, \beta)$-fair solution, as claimed.
\end{proof}






\subsection{Proof of Claim~\ref{claim3}}
\begin{claim}[Restatement of Claim~\ref{claim3}]
Let $r$ be any round in Algorithm~\ref{alg:fair_cost} such that $r\leq k\cdot\log_{\frac{2+2\varepsilon}{2+\varepsilon}} k - k$, and suppose the algorithm is running for $k$ additional rounds beyond $r$. Then:
\begin{align*}
\hat c_{H_{r + k}}^{(r + k)} \leq \left( \frac{2+\varepsilon}{2+2\varepsilon} \right) \cdot \hat c_{H_r}^{(r)}.
\end{align*}
\end{claim}

\begin{proof}
\label{app:claim3}
Suppose we run Algorithm~\ref{alg:fair_cost} for $k+1$ consecutive rounds, starting from round $r$. In each round, exactly one server is designated as the highest-cost server and gets swapped. Since there are only $k$ servers but $k+1$ rounds, the pigeonhole principle implies that at least one server is selected as the highest-cost server in two different rounds.

Formally, there exists a server $b \in [k]$ and two distinct rounds $r + m_1$ and $r + m_2$ with $0\leq m_1 < m_2 \leq k$, such that $
H_{r + m_1} = H_{r + m_2} = b.$
That is, server $b$ is designated as the highest-cost server in both rounds $r + m_1$ and $r + m_2$.
We distinguish two cases:

\smallskip

   \noindent \textbf{Case (I).}The server $b$ is never swapped as a low-cost server during the interval between rounds $r + m_1$ and $r + m_2$. 
   
   \smallskip

   \noindent \textbf{Case (II).} Server $b$ may be swapped as a low-cost server during that interval.

We first consider Case (I). Since $b$ was designated as the highest-cost server in both rounds $r + m_1$ and $r + m_2$ and was not swapped as a low-cost server in between, its cost remained unchanged between rounds $r + m_1 + 1$ and $r + m_2$. We now apply the previous claims step-by-step: 

\begin{align}
\hat c_{H_{r + k}}^{(r + k)} 
&\leq \hat c_{H_{r + m_2}}^{(r + m_2)} 
&&\text{(by Claim~\ref{claim2}, decreasing max cost)} \nonumber \\
&= \hat c_b^{(r + m_2)} 
&&\text{(since $H_{r + m_2} = b$)} \nonumber \\
&= \hat c_b^{(r + m_1 + 1)} 
&&\text{(no swaps involving $b$ between $r + m_1 + 1$ and $r + m_2$)} \nonumber \\
&\leq \left( \frac{2+\varepsilon}{2+2\varepsilon}  \right) \cdot \hat c_b^{(r + m_1)} 
&&\text{(by Claim~\ref{claim1}, applied at $r + m_1$)} \nonumber \\
&= \left( \frac{2+\varepsilon}{2+2\varepsilon}  \right) \cdot \hat c_{H_{r + m_1}}^{(r + m_1)} 
&&\text{(since $H_{r + m_1} = b$)} \nonumber \\
&\leq \left( \frac{2+\varepsilon}{2+2\varepsilon} \right) \cdot \hat c_{H_t}^{(r)} 
&&\text{(again by Claim~\ref{claim2})}. \nonumber
\end{align}

Now in case (II) Suppose server $b$ is swapped as the low-cost server at least once between rounds $r + m_1$ and $r + m_2$. Let $t + m_3$ denote the last such round where this occurs, with $m_1 < m_3 < m_2$.
Then, by Claim~\ref{claim2} (non-increasing maximum cost), we have:
\begin{align}
\hat c_{H_{r + k }}^{(r + k )} 
&\leq \hat c_{H_{r + m_2}}^{(r + m_2)} \label{eq:eqk11} \\
&= \hat c_b^{(r + m_2)} = \hat c_b^{(r + m_3 + 1)} \nonumber \\
&\leq \left( \frac{2+\varepsilon}{2+2\varepsilon}  \right) \cdot \hat c_{H_{r + m_3}}^{(r + m_3)} 
&&\text{(by Claim~\ref{claim1})} \label{eq:eqk112} \\
&\leq \left( \frac{2+\varepsilon}{2+2\varepsilon}  \right) \cdot \hat c_{H_r}^{(r)} 
&&\text{(by Claim~\ref{claim2})}. \label{eq:eqk113}
\end{align}
Thus, in both Case (I) and Case (II), the maximum server cost after $k$ rounds decreases by at least a factor of $\left( \frac{2+\varepsilon}{2+2\varepsilon}  \right)$, completing the proof of Claim~\ref{claim3}.
\end{proof}


``

\subsection{Proof of Theorem~\ref{thm:randomized}}

\begin{theorem}[Restatement of Theorem~\ref{thm:randomized}]
    Let $A$ be any online $k$-server algorithm and let $\sigma$ be any request sequence. Then, Algorithm~\ref{alg:randomized_fairness} with parameter $\gamma > 0$ satisfies the following:

\begin{itemize}
    \item The expected cost incurred by each server is
    \[
    \frac{\text{cost}(A, \sigma)}{k} + O\left( \text{cost}(A, \sigma)^{1/(\gamma+1)} \cdot \diam \right).
    \]
    
    \item More importantly, for any $\varepsilon > 0$, with probability at least
    \[
    1 - k \cdot \exp\left( -\Omega \left ( \frac{\varepsilon^2}{\gamma \cdot k} \cdot \text{cost}(A, \sigma)^{1/(\gamma+1)} \right) \right),
    \]
    each server incurs cost at most
    \[
    (1 + \varepsilon) \cdot \frac{\text{cost}(A, \sigma)}{k} + O\left( \text{cost}(A, \sigma)^{1/(\gamma+1)} \cdot \diam \right).
    \]
\end{itemize}
\end{theorem}

\begin{proof}
\label{apx:randomized}
For $j \in [m]$ let $\phi_j = j^\gamma$ denote the length of phase $j$. Let $X_{ij}$ denote the cost assigned to server $i$ during phase $j$.We define the total cost assigned to server $i \in [k]$ without considering swaps as:
\[
c'_i = \sum_{j=1}^m X_{ij}.
\]
For $j \in [m]$, we define $U_j$ to be the total cost of each phase(look at the Line \ref{line:update_u}), which is sum of all cost assigned to servers: 
\[
 U_j = \sum_{i=1}^k X_{ij} \leq \phi_j + \diam.
\]
Note that in each phase $j$, we maintain the invariant that $U_j < \phi_j$ until the phase ends, and the final request may increase the phase cost by at most $\diam$.
For the final phase, we have $U_m \le \phi_m$, since no further phase change occurs.

\noindent The expected cost $c_i$ for each server $i$ is $c'_i$ plus its swap costs, which are at most $m \cdot \diam$:
\[
\mathbb{E}[c_i] \leq E[c'_i] + m\cdot \diam = \mathbb{E}[X_{i1}] + \mathbb{E}[X_{i2}] + \cdots + \mathbb{E}[X_{im}] + m\cdot \diam ,
\]
Since the assignment of server IDs within each phase is uniform at random we have:
\[
\mathbb{E}[X_{ij}] = \frac{\sum_{i=1}^k X_{ij}}{k} =\frac{U_j}{k}.
\]
So, recalling that the total cost incurred by algorithm $A$ is $\text{cost}(A) = \sum_{j=1}^{m} U_j$, the expected total cost for server $i$ is:
\begin{equation}
\label{eq:r1}
\mathbb{E}[c_i] \leq \sum_{j=1}^{m} \frac{U_j}{k} + m\cdot \diam = \frac{\text{cost}(A)}{k} + m\cdot \diam.
\end{equation}

To relate the number of phases $m$ to the total cost $\cost(A,\sigma)$, we bound the cumulative cost incurred across phases.
By construction of Algorithm~\ref{alg:randomized_fairness}, in each phase $j < m$ the phase ends only after the accumulated cost reaches the threshold $\phi_j$, and serving the final request may increase the phase cost by at most $\diam$.
Thus, phase $j$ incurs cost at most $\phi_j + \diam$.
For the final phase, no further phase change occurs, and therefore $U_m \le \phi_m$.
Consequently, we have

\begin{equation}
    \sum_{j=1}^{m-1} \phi_j\leq\text{cost}(A) \leq \sum_{j=1}^{m} (\phi_j + \diam) = \sum_{j=1}^{m} j^\gamma + m \cdot \diam. \label{eq:phase_cost_bound}
\end{equation}
We use the standard bound based on integral approximation:
\[
\int_1^m x^\gamma \, dx \leq \sum_{j=1}^m j^\gamma \leq \int_1^{m+1} x^\gamma \, dx,
\]
which yields 
\begin{equation}
\label{eq:bound1}
\sum_{j=1}^m j^\gamma = \Theta \left (\frac{m^{\gamma+1}}{\gamma + 1} \right ) \quad \text{for } \gamma > -1.
\end{equation}
Using equation~\eqref{eq:phase_cost_bound} and equation~\eqref{eq:bound1} for a fixed constant  $\gamma>0$, we conclude that:
\[
\text{cost}(A) = \Theta \left (\frac{m^{\gamma+1}}{\gamma + 1} \right ),
\quad \text{so} \quad m = \Theta\left( \text{cost}(A)^{\frac{1}{\gamma+1}} \right).
\]
Plugging this back into the Equation~\eqref{eq:r1} for the expected cost of server $i$, we obtain:
\[
\mathbb{E}[c_i] \leq \frac{\text{cost}(A)}{k} + O\left( \text{cost}(A)^{\frac{1}{\gamma+1}} \cdot \diam \right).
\]
%
%
So, for any fixed constat $\gamma >0$, the algorithm is $(1,O( \text{cost}(A)^{\frac{1}{\gamma+1}} \cdot \diam )$-fair in expectation (i.e., ex-ante fair). 

We now show that the algorithm is also fair with high probability (i.e., ex-post fair), which is the more interesting result. We will prove concentration of the random costs of the servers using Bernstein's inequality. 
\begin{fact}[Bernstein’s Inequality]
\label{thm:ber}
Let $X_1, \dots, X_n$ be independent random variables and assume that each satisfies $|X_i| \leq M$ almost surely. Let $S_n = \sum_{i=1}^n X_i$ and let $\sigma^2 = \sum_{i=1}^n \operatorname{Var}(X_i)$. Then for any $\varepsilon > 0$,
\[
\Pr\left( S_n - \sum_{i=1}^n \E[X_i]  \geq \varepsilon \right) \leq \exp\left( -\frac{\varepsilon^2}{2\sigma^2 + \frac{2}{3}M\varepsilon} \right).
\]
\end{fact}

To avoid complications from phase overflow (where the final request in a phase may cause the total cost to exceed the phase threshold $\phi_j$), we define new random variables $Y_{ij}$ as follows: let $Y_{ij}$ denote the cost assigned to server $i$ during phase $j$, counting only up to the point when the phase threshold $\phi_j$ is reached. That is,
\[
Y_{ij} \leq X_{ij} \leq Y_{ij} + \diam,
\]
due to at most one request exceeding the threshold per phase.

Now we define $D_i$:
\[
D_i := \sum_{j=1}^m Y_{ij},
\]
which represents the total cost assigned to server $i$ excluding the phase overflow terms and swapping costs. We now analyze the concentration of $D_i$ using Bernstein's inequality.
\[
\operatorname{Var}(Y_{ij}) = \mathbb{E}[Y_{ij}^2] - \mathbb{E}[Y_{ij}]^2 \leq \phi_j^2 \left( \frac{1}{k} - \frac{1}{k^2} \right) < \frac{\phi_j^2}{k}.
\]

Next step Define
\[
Z \;\coloneqq\; \sum_{j=1}^{m-1} j^\gamma \;+\; U_m .
\]
Then, since $\E[Y_{ij}] = j^\gamma/k$ for $j<m$ and $\E[Y_{im}] = U_m/k$, we have
\[
\E[D_i] \;=\; \sum_{j=1}^m \E[Y_{ij}] \;=\; \frac{1}{k}\left(\sum_{j=1}^{m-1} j^\gamma + U_m\right)
\;=\; \frac{Z}{k}.
\]


We now apply Bernstein's inequality (Fact~\ref{thm:ber}) to bound the probability that $D_i$ (which also bounds $c_i$ up to an additive $2m \cdot \diam$) deviates significantly above its expectation:
\[
\Pr\left(D_i-\mathbb{E}[D_i] \geq \varepsilon \cdot \frac{Z_i}{k} \right) =\Pr\left(D_i \geq (1+\varepsilon) \cdot \frac{Z_i}{k} \right)
\leq \exp\left( \frac{-\varepsilon^2 \cdot Z_i^2 / k^2}{2\sigma^2 + \frac{2}{3} M \cdot \varepsilon \cdot Z_i/k} \right),
\]
where:
 \( M = \phi_m = m^\gamma \) is the maximum possible value of any $Y_{ij}$,
and the total variance satisfies:
\[
\sigma^2 \leq \sum_{j=1}^m \operatorname{Var}(Y_{ij}) \leq \frac{1}{k} \left( \sum_{j=1}^{m-1} \phi_j^2 + U_m^2 \right)\leq \frac{1}{k} \left( \sum_{j=1}^{m} \phi_j^2  \right),
\]
where \( U_m \leq \phi_m \) is the cost of the final phase and $\varepsilon$ is artbitray fixed constant.

To complete the proof, we verify that both terms in the denominator of the Bernstein inequality are asymptotically dominated by the numerator.

\noindent Recall that
\[
Z_i = \sum_{j=1}^{m-1} \phi_j + U_m .
\]
Since $\phi_j = j^\gamma$ and $U_m \le \phi_m$, we have
\[
Z_i
= \Theta\!\left( \sum_{j=1}^{m} j^\gamma \right)
= \Theta\!\left( \frac{m^{\gamma+1}}{\gamma+1} \right).
\]
It follows that
\[
Z_i^2
= \Theta\!\left( \frac{m^{2\gamma+2}}{(\gamma+1)^2} \right),
\qquad
\frac{Z_i^2}{k^2}
= \Theta\!\left( \frac{m^{2\gamma+2}}{k^2(\gamma+1)^2} \right).
\]
Next, recall that
\[
\sigma^2 = \sum_{j=1}^m \operatorname{Var}(Y_{ij})
\le \frac{1}{k} \sum_{j=1}^m \phi_j^2 .
\]
Since $\phi_j = j^\gamma$, this yields
\[
\sigma^2
= O\!\left( \frac{1}{k} \sum_{j=1}^m j^{2\gamma} \right)
= O\!\left( \frac{m^{2\gamma+1}}{k(2\gamma+1)} \right).
\]
For Bernstein’s inequality, the random variables are uniformly bounded as
\[
|Y_{ij}| \le M \coloneqq \phi_m .
\]
Therefore,
\[
M Z_i
= \phi_m Z_i
= \Theta\!\left( m^\gamma \cdot \frac{m^{\gamma+1}}{\gamma+1} \right)
= \Theta\!\left( \frac{m^{2\gamma+1}}{\gamma+1} \right),
\]
and consequently,
\[
\frac{2}{3}\cdot \frac{M Z_i}{k}
= \Theta\!\left( \frac{m^{2\gamma+1}}{k(\gamma+1)} \right).
\]

Plugging the bounds on $\E[D_i]$, $\sigma^2$, and $M=\phi_m$ derived above into
Bernstein’s inequality, and using $0<\varepsilon\le 1$, we obtain
\[
\Pr\left(D_i \ge (1+\varepsilon)\E[D_i]\right)
\le
\exp\!\left(-\Omega\!\left(\frac{\varepsilon^2}{\gamma}\cdot\frac{m}{k}\right)\right).
\]

Where \( m = \text{cost}(A)^{1/(\gamma+1)} \). Now, for fixed \( \varepsilon > 0 \), we conclude that with probability at least
\[
1 - \exp\left( -\varepsilon^2 \cdot \frac{1}{\gamma} \cdot \Omega\left( \text{cost}(A)^{1/(\gamma+1)}/k \right) \right),
\]
we have
\[
D_i < (1 + \varepsilon) \cdot \frac{\text{cost}(A)}{k}.
\]
Since the actual cost satisfies \( c_i \leq D_i + 2m \cdot \diam \), we have that, with the same probability,
\[
c_i \leq (1 + \varepsilon) \cdot \frac{\text{cost}(A)}{k} + O\left( \text{cost}(A)^{1/(\gamma+1)} \cdot \diam \right).
\]
Finally, to ensure this holds simultaneously for all servers, we apply a union bound. With probability at least
\[
1 - k \cdot \exp\left( -\varepsilon^2 \cdot \frac{1}{\gamma} \cdot \Omega\left( \text{cost}(A)^{1/(\gamma+1)}/k \right) \right),
\]
we have
\[
\forall i \in [k]: \quad
c_i \leq (1 + \varepsilon) \cdot \frac{\text{cost}(A)}{k} + O\left( \text{cost}(A)^{1/(\gamma+1)} \cdot \diam \right),
\]
as claimed.

\end{proof}

\begin{remark}[Choosing a target failure probability]
Fix $\varepsilon,\gamma>0$ and a target failure probability $\zeta\in(0,1)$.
Theorem~\ref{thm:randomized} guarantees failure probability at most $\zeta$ provided
\[
k\cdot \exp\!\left(-\Omega\!\left(\frac{\varepsilon^2}{\gamma k}\cdot \cost(A,\sigma)^{1/(\gamma+1)}\right)\right)\le \zeta.
\]
Equivalently, it suffices that
\[
\cost(A,\sigma)\ \ge\ \left(\Theta\!\left(\frac{\gamma k}{\varepsilon^2}\cdot \log\frac{k}{\zeta}\right)\right)^{\gamma+1}.
\]
\end{remark}

\subsection{Proof of Theorem~\ref{thm:additive-fairness}}
\begin{theorem}[Restatement of Theorem~\ref{thm:additive-fairness}]
Let $\sigma$ be any request sequence over a line metric with diameter $\diam$, and let $c_i$ denote the total cost incurred by server $i$ when serving $\sigma$ using the Double Coverage Algorithm (DCA). Then for any $i, j \in [k]$,
\[
|c_i - c_j| = O(k \cdot \diam).
\]
Therefore, DCA in line metric is $(1, O(k \cdot \diam))$-fair. 
\end{theorem}
\begin{proof}\label{apx:dca_on_line}
We consider the Double Coverage Algorithm (DCA) on a line metric.
Since DCA never allows servers to pass each other, we can index the servers
from left to right as $1,2,\dots,k$ according to their initial positions, and
this ordering is preserved throughout the execution.

For each server $i$, let $R_i$ denote the total distance it moves to the right
and $L_i$ the total distance it moves to the left. The total cost incurred by
server $i$ is therefore
\[
c_i = R_i + L_i.
\]

Under DCA on a line, a server moves to the right if and only if a request
arrives between servers $i$ and $i+1$, in which case both servers move toward
the request at equal speed until it is served. Similarly, server $i$ moves to
the left if and only if a request arrives between servers $i-1$ and $i$.
Consequently, every unit of rightward movement by server $i$ corresponds to an
equal unit of leftward movement by server $i+1$, and hence
\[
R_i = L_{i+1}.
\]

Define the net displacement of server $i$ as
\[
d_i \coloneqq R_i - L_i.
\]
Since all server movement occurs within a line segment of diameter $\diam$, we
have $|d_i| \le \diam$ for all $i$.

We now relate the costs of adjacent servers. Using the above definitions,
\begin{align*}
c_i &= R_i + L_i = 2R_i - d_i, \\
c_{i+1} &= R_{i+1} + L_{i+1} = 2L_{i+1} + d_{i+1}.
\end{align*}
Using $R_i = L_{i+1}$, we obtain
\[
c_{i+1} - c_i = d_{i+1} + d_i.
\]

For any $j \ge 1$, the difference between servers $i$ and $i+j$ can be written
as a telescoping sum:
\[
c_{i+j} - c_i
= \sum_{\ell=0}^{j-1} (c_{i+\ell+1} - c_{i+\ell})
= d_{i+j} + 2d_{i+j-1} + \cdots + 2d_{i+1} + d_i.
\]
Since $|d_\ell| \le \diam$ for all $\ell$, we conclude that
\[
|c_{i+j} - c_i| \le 2j \cdot \diam \le 2(k-1)\cdot \diam.
\]

Thus, the difference in cost between any two servers is bounded by
$O(k\cdot \diam)$, which implies that DCA is $(1, O(k\cdot \diam))$-fair on line
metrics(additively fair algorithm )
\end{proof}

\subsection{Omitted Details of Theorem~\ref{thm:multiplicatively_fair}}
The key idea of the proof is to relate the individual server costs to how the
distance between the two servers evolves over time. For $k=2$, every movement
either increases or decreases the distance between the servers, and this change
is captured precisely by the divergence and convergence terms. Lemma~\ref{lemma:DCATwo}
formalizes this intuition by bounding the imbalance between divergence and
convergence using the diameter of the metric.
\begin{lemma}
For any input $\sigma$ and any server $i$ in the $k=2$ setting, we have
\begin{equation}
   -\diam \leq div_i(\sigma) + div_{i'}(\sigma) - 2\,con_i(\sigma) \leq \diam,
\end{equation}
where $\diam$ is the diameter of the metric. \label{lemma:DCATwo}\label{lem:DCAdivConv}
\end{lemma}

\begin{proof}\label{apx:dcadivconv}
Let $g$ denote the initial distance between servers $i$ and $i'$. We track the distance between them at each step. Each time exactly one server diverges from the other, the distance increases by the amount of that divergence; each time both converge, it decreases by twice the amount each of them converges. Hence,
\begin{equation*}
    distance(i,i') \;=\; g + div_i(\sigma) + div_{i'}(\sigma) - 2\,con_i(\sigma).
\end{equation*}
At the end of the sequence, the distance is upper bounded by $\diam$, so we obtain
\begin{equation*}
    0 \;\leq\; g + div_i(\sigma) + div_{i'}(\sigma) - 2\,con_i(\sigma) \;\leq\; \diam.
\end{equation*}
And since $g \le \diam$, it follows that
\begin{equation*}
    -\diam \;\leq\; div_i(\sigma) + div_{i'}(\sigma) - 2\,con_i(\sigma) \;\leq\; \diam.
\end{equation*}
\end{proof}

\subsection{Proof of Theorem~\ref{thm:DCA_additive_fair}}
\begin{theorem}[Restatement of Theorem~\ref{thm:DCA_additive_fair}]
There exists a tree metric and a request sequence $\sigma$ such that, when the Double Coverage Algorithm (DCA) is run on $\sigma$, some server incurs cost $\Omega(k \cdot \OPT(\sigma))$, where $\OPT(\sigma)$ denotes the cost of the offline optimal solution.
\end{theorem}
\begin{proof}
\label{apx:DCA_not_fair}

We first describe the hard tree metric, and then the hard request sequence on this metric.

\smallskip
\begin{figure}[t]
    \centering
    \includegraphics[width=0.8\linewidth]{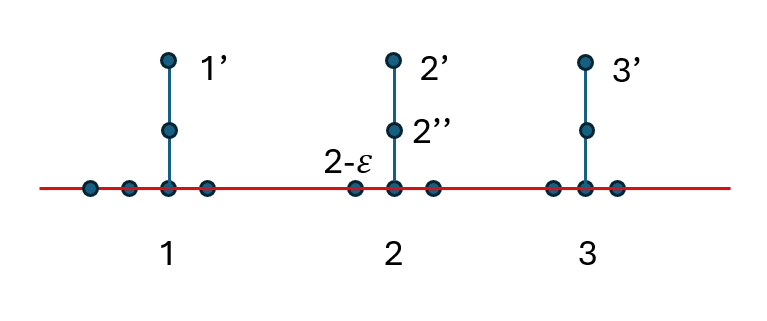}
    \caption{Illustration of the tree metric used in the proof of Theorem~\ref{thm:DCA_additive_fair}.
The horizontal red line represents the spine of the metric, with major nodes labeled
$1,2,3$. Each major node has an attached minor node at distance $\varepsilon$ which is connected to a leaf at distance$\varepsilon$.
And also we have nodes at $i\pm\varepsilon$ and one node at $1-2\varepsilon$.
The construction highlights the geometric structure of the metric; the request sequence
used in the proof exploits this structure to create large imbalance in server costs.}
    \label{fig:dca-hard-instance}
\end{figure}
\noindent\textbf{Metric.}
Fix a small $\varepsilon$.
Let the \emph{spine} be a path with vertices $1,2,\dots,k{+}2$ and unit edges $d(i,i{+}1)=1$.
Subdivide each spine edge near both endpoints by adding the vertices $i-\varepsilon$ on $(i{-}1,i)$ and $i+\varepsilon$ on $(i,i{+}1)$
(so $d(i,i-\varepsilon)=d(i,i+\varepsilon)=\varepsilon$ along the spine).
For each $i\in\{1,\dots,k{+}2\}$ attach a length-$2\varepsilon$ leaf at $i$
(a two-edge path of two $\varepsilon$-edges).
Let $i'$ denote the far endpoint of this leaf, and let $i''$ denote the midpoint
of the leaf, i.e., the point at distance $\varepsilon$ from $i$.

Finally, extend the spine by adding two additional points:
a point $1-2\varepsilon$ at distance $2\varepsilon$ to the left of vertex~$1$,
and a point $k+2+2\varepsilon$ at distance $2\varepsilon$ to the right of vertex~$k+2$. 

\noindent\textbf{Initial configuration.}
Place servers at $1-2\varepsilon,3,4,5,\dots,k{+}1$ (no server at $2$). We focus on the leftmost server $s_1$ (starting at $1-\varepsilon$). 

\noindent\textbf{Request sequence (one right-sweep; short form).}
We issue the following requests in order:
\[
2',\ 3',\ 4',\ \dots,\ k',\ \ k{+}2{+}2\varepsilon,\ \ k{+}2{+}\varepsilon.
\]
Then we \emph{mirror} the pattern from the right side (replace $i$ by $k{+}3{-}i$ and $\varepsilon$ by $-\varepsilon$) to perform a left-sweep and restore the initial pattern; repeating yields an arbitrarily long sequence.

\smallskip

\noindent\textbf{DCA dynamics after each request.}
We now describe the behavior of DCA step by step:
\textbf{Request at $2'$.}
Initially $s_1=1-\varepsilon$, $s_2=3$, $s_3=4$, etc. 
The servers unblocked on the unique paths to $2'$ are $s_1$ and $s_2$, so both start moving.
After time $1$, $s_1$ reaches $2-\varepsilon$ while $s_2$ reaches $2$; at this moment $s_1$ becomes \emph{blocked}
(because $s_2$ lies on its path to $2'$) and stops at $2-\varepsilon$.
Then $s_2$ continues alone along edge $(2,2')$ for distance $2\varepsilon$ and serves the request at $2'$.
Thus:
\[
s_1(1)=2-\varepsilon, \quad s_2(1)=2', \quad s_3(1)=4, \ s_4(1)=5, \ \dots
\]
\textbf{Request at $3'$.}
All servers whose unique path to $3'$ is unobstructed begin moving: $s_1$, $s_2$, and $s_3$.

\smallskip

\noindent\emph{Blocking events and timings.}
At time $t=\varepsilon$, $s_1$ reaches node $2$, while $s_2$ (coming along the leaf) reaches the midpoint of edge $(2,2')$ vertex $2''$.
Since $s_1$ lies on $s_2$’s path to $3'$, $s_2$ becomes \emph{blocked} at this midpoint and stops there.
$s_1$ continues toward $3$.

At time $t=1$, $s_1$ is at $3-\varepsilon$ and $s_3$ reaches node $3$.
Because $s_3$ is now on $s_1$’s path to $3'$, $s_1$ becomes \emph{blocked} at $3-\varepsilon$ and stops.
$s_3$ continues alone along the leaf from $3$ to $3'$ and serves the request at time $t=1+2\varepsilon$.

\emph{State after serving $3'$.}
\[
s_1(2) = 3-\varepsilon,\qquad s_2(2) = \text{midpoint of }(2,2'),\qquad s_3(2) = 3',\qquad s_4(2)=5,\ \dots
\]
\textbf{Continuing for $4',5',\dots,k'$.}
By induction, at each step, the designated leftmost server $s_1$ advances exactly one unit further
along the spine, stopping at position $i+1-\varepsilon$ after moving toward request~$i'$. 
At the same time, the server initially located at $i{+}1$ is diverted to the leaf node $i'$ and then will move to midpoint $(i-1)''$. 
Thus, after processing $i'=5',6',\dots,k'$, we have
\[
s_1(k-1) = k-\varepsilon,\qquad
s_j(k-1) = j''\ \text{for } j=2,3,\dots,k-1,\qquad
s_k(k-1) = k'
\]
In total, $s_1$ has moved $k-1$ units along the spine,
while each intermediate server $s_j$ ($2\le j\le k$) has made exactly one excursion of length $1+2\varepsilon$
when it was the right-hand partner for request~$j'$ and $\varepsilon$ when $s_1$ went for the next request.

\smallskip

\noindent\textbf{Request at $k{+}2{+}2\varepsilon$.}
Both $s_k$ (at $k'$) and $s_1$ (at $k-\varepsilon$) are unblocked toward this leaf and begin moving.
After distance $\varepsilon$, $s_k$ reaches $k''$ and is then blocked by $s_1$ at node $k$.
From that point, only $s_1$ continues moving along the spine and serves the request at $k{+}2{+}2\varepsilon$.
Thus $s_1$ incurs an additional distance of $2+\varepsilon$ in this step, while $s_k$ moves only $\varepsilon$.

\smallskip

\noindent\textbf{Request at $k{+}2{+}\varepsilon$.}
At this moment all servers are unblocked toward the spine.
Each parked server at $j''$ moves distance $\varepsilon$ back to its spine node $j$,
for example server $s_k$ moves distance $\varepsilon$ from $k''$ back to $k$,
and $s_1$ moves distance $\varepsilon$ back from $k{+}2{+}2\varepsilon$ to $k{+}2{+}\varepsilon$.
Hence every server shifts by exactly $\varepsilon$ and all servers return to the spine. All servers are in a similar postion to their initial positions.

Let $f$ be the left--right mirror of the spine defined by
\[
f(x) \coloneqq (k+3) - x,
\]
extended naturally to the subdivided points and leaves by mapping
$i' \mapsto (k+3-i)'$, $i'' \mapsto (k+3-i)''$, and
$(i\pm \varepsilon) \mapsto (k+3-i) \mp \varepsilon.
$
Let $\mathbf{s}{(0)}=(s_1{(0)},\dots,s_k{(0)})$ denote the initial server
configuration and let $\mathbf{s}^{(\mathrm{R})}$ denote the configuration at the
end of the right-sweep (after serving request $k+2+\varepsilon$).
Then $\mathbf{s}({k+1})$ has the same pattern as $\mathbf{s}^{(0)}$ under
mirror symmetry: namely,
\[
\bigl(s_1{(\mathrm{k+1})},\,s_2{(\mathrm{k+1})},\,\dots,\,s_k{(\mathrm{k+1})}\bigr)
=
\bigl(f(s_1{(0)}),\,f(s_2{(0)}),\,\dots,\,f(s_k{(0)})\bigr).
\]
Therefore, issuing the mirrored request sequence (replacing $x$ by $f(x)$, i.e.,
replacing $i$ by $k+3-i$ and $\varepsilon$ by $-\varepsilon$) produces the
symmetric left-sweep and returns the servers to the original configuration.


Once this sweep is complete, we reverse the process to bring $s_1$ back to the leftmost position, and repeat. 
Assuming $r$ repetitions, the cost of the $s_1$ server is
\begin{equation}
    c_1 = r \cdot (k+1+4\varepsilon) = \Theta(rk)
\end{equation}
The total cost incurred by all other servers is
\begin{equation}
    \sum c_i = r \cdot(1+3\varepsilon)\cdot k = \Theta(rk).
\end{equation}
%
%

Moreover, the total cost of the optimal solution is less than $r\cdot(2+8\varepsilon)+((k-2)\cdot 2\varepsilon)=\Theta(r)$, since this optimal solution could move only $s_1$ to do 2' then every server can serve the request at its own leaf with cost $2\varepsilon$, and server $k$ can serve the request at $k+2+2\varepsilon$ and return to $k+2+\varepsilon$ position. So the total cost will be the movement of $1+3\varepsilon$ for $s_1$,  $2\varepsilon$ for each intermediate server, and $2+5\varepsilon$ for $s_k$ in each sweep(expect the first sweep which you have to do $3\varepsilon$ not 5). 

Thus, we have that $s_1$ pays cost $\Omega(k \cdot \OPT(\sigma))$ as claimed.
%


\end{proof}

\end{document}